\documentclass[11pt]{llncs}
\usepackage{amsfonts,amsmath,amssymb,latexsym,comment}
\ifx\pdftexversion\undefined
  \usepackage[colorlinks,linkcolor=black,filecolor=black,citecolor=black,urlcolor=black,pdfstartview=FitH]{hyperref}
\else
   \usepackage[colorlinks,linkcolor=blue,filecolor=blue,citecolor=blue,urlcolor=blue,pdfstartview=FitH]{hyperref}

\fi
\usepackage{graphicx}
\usepackage{xspace}
\usepackage{subfigure}
\usepackage{lscape}
\newcommand{\ie}{\emph{i.e.,\ }}
\newcommand{\Ie}{\emph{I.e.,\ }}

\newcommand{\namedref}[2]{\hyperref[#2]{#1~\ref*{#2}}}
\newcommand{\sectionref}[1]{\namedref{Section}{#1}}

\newcommand{\theoremref}[1]{\namedref{Theorem}{#1}}

\newcommand{\figureref}[1]{\namedref{Figure}{#1}}

\newcommand{\lemmaref}[1]{\namedref{Lemma}{#1}}
\newcommand{\remarkref}[1]{\namedref{Remark}{#1}}
\newcommand{\definitionref}[1]{\namedref{Definition}{#1}}

\newcommand{\corollaryref}[1]{\namedref{Corollary}{#1}}

\newcommand{\hide}[1]{}
\newcommand{\hideForShortVersion}[1]{#1}
\newcommand{\hideForLongVersion}[1]{}

\newcommand{\dd}[1]{\textbf{\color{red}[[[#1]]]}}

\newcommand{\ezra}[1]{\textbf{\color{blue}
[[[EZRA: #1]]]}}
\renewcommand{\ezra}[1]{}

\newcommand{\tb}{\makebox[0.4cm]{}}
\newcommand{\due}{\makebox[0.8cm]{}}
\newcommand{\tre}{\makebox[1.2cm]{}}

\newcommand\A{\mathcal{A}}
\newcommand\TS{\mathcal{TS}}
\newcommand\I{\mathcal{I}}

\newcommand\C{\mathcal{C}}
\newcommand\R{\mathcal{R}}
\newcommand\T{\mathcal{T}}
\newcommand\p{\mathcal{P}}
\newcommand\V{\mathcal{V}}
\newcommand\F{\mathcal{F}}

\newsavebox{\theorembox}
\newsavebox{\lemmabox}
\newsavebox{\conjecturebox}
\newsavebox{\claimbox}
\newsavebox{\factbox}
\newsavebox{\corollarybox}
\newsavebox{\definitionbox}
\newsavebox{\propositionbox}
\newsavebox{\examplebox}
\savebox{\theorembox}{\bf Theorem}

\savebox{\lemmabox}{\bf Lemma}

\savebox{\conjecturebox}{\bf Conjecture}

\savebox{\claimbox}{\bf Claim} \savebox{\factbox}{\bf Fact}

\savebox{\corollarybox}{\bf Corollary}
\savebox{\definitionbox}{\bf Definition}
\savebox{\propositionbox}{\bf Proposition}

\savebox{\examplebox}{\bf Example}

\newtheorem{notation}{{\sc Notation}\rm }[section]

\newcommand{\ignore}[1]{}
\newcommand{\byzantine}[0]{\mbox{\emph{Byzantin}\hspace{-0.15em}\emph{e}}\xspace}

\newcommand{\clockAlg}{\textsc{Async-Clock}\xspace}
\newcommand{\enmasse}{\textsc{EnMasse}\xspace}

\def\squarebox#1{\hbox to #1{\hfill\vbox to #1{\vfill}}}

\newcommand{\countv}{\mbox{\textsl{count}}\xspace}
\newcommand{\pass}{\mbox{\textsl{pass}}\xspace}
\newcommand{\myval}{\mbox{\textsl{my\_val}}\xspace}
\newcommand{\modk}[2]{#1\oplus_k#2}
\newcommand{\linenumber}[1]{{\tt #1}}
\newcommand{\lineref}[1]{Line~\linenumber{#1}}
\usepackage{color}

\begin{document}

\title{A Fault-Resistant Asynchronous Clock Function}
\author{Ezra N. Hoch, Michael Ben-Or, Danny Dolev}

\institute{
The Hebrew University of Jerusalem\\
\texttt{\{ezraho,benor,dolev\}@cs.huji.ac.il}
}

\maketitle

\begin{abstract}

Consider an asynchronous network in
a shared-memory environment consisting of $n$ nodes.
Assume that up to $f$ of the nodes might be
\byzantine ($n > 12f$), where the adversary is full-information and dynamic (sometimes called adaptive). In addition, the non-\byzantine nodes may undergo transient failures. Nodes advance in atomic steps, which consist of reading
all registers, performing some calculation and writing to all
registers.

The three main contributions of the paper are: first, the clock-function problem is defined, which is a generalization of the clock synchronization problem. This generalization encapsulates previous clock synchronization problem definitions while extending them to  the current paper's model. Second, a randomized asynchronous self-stabilizing \byzantine tolerant clock synchronization algorithm is presented.

In the construction of the clock synchronization algorithm, a building block that ensures different nodes advance at similar rates is developed. This feature is the third contribution of the paper. It is self-stabilizing and \byzantine tolerant and can be used as a building block for different algorithms that operate in an asynchronous self-stabilizing \byzantine model.

The convergence time of the presented algorithm is exponential. Observe that in the asynchronous setting the best known full-information dynamic \byzantine agreement also has
an expected exponential convergence time.

\hide{
In the synchronous setting, self-stabilizing \byzantine tolerant clock synchronization is equivalent to \byzantine agreement. In the asynchronous setting the best known full-information dynamic \byzantine agreement has
expected exponential convergence time. It is not known if the synchronous equivalence between clock synchronization and \byzantine agreement transfers to the asynchronous setting. While the clock synchronization algorithm presented in this work has expected exponential convergence time, it is possible that an improvement in its running time will lead
to a better convergence of asynchronous \byzantine agreement. That is, if asynchronous \byzantine agreement can be achieved using asynchronous self-stabilizing \byzantine clock synchronization, then improving the current work's
convergence will require usage of new methods (as it will imply an improvement to the best known asynchronous full-information dynamic \byzantine agreement). \dd{you didn't mention that we assume full information.\ezra{first pgph, line 2,3}  i wouldn't put the las sentence as is - i will tome it down saying that ours is also exponential. i will remove the exponential from the end of the 2nd par}\ezra{is it better?}}
\end{abstract}

\hide{
\vspace{3mm}
\hrule
\hrule
\vspace{3mm}
\mbox{\tb \  \tre}Shared memory.\\
\\
Contact author: Ezra Hoch\\
Email: ezra.hoch@gmail.com\\
Tel: +972-2-658-5770;
Fax: +972 2 658 5727\\
\\
Postal Address: The Selim and Rachel Benin\\
School of Computer Science and Engineering,\\
The Hebrew University of Jerusalem,\\
Ross Building,\\
Givat Ram Campus,\\
91904 Jerusalem, Israel\\
\\
Submitted as a regular paper.
\hrule
\hrule


\newpage

} 

\section{Introduction}
When tackling problems in distributed systems, there are many previously developed building blocks that assist in solving the problem. Some of these building blocks allow one to design a solution under ``easy'' assumptions, then automatically transform them to a more realistic environment. For example,  it is easier to construct an algorithm in the synchronous model, then add an underlying synchronizer (see \cite{4227}) to adapt the solution to an asynchronous model. Similarly, developing a self-stabilizing algorithm can be challenging; instead, one can develop a non-self-stabilizing algorithm, and use a stabilizer (\cite{Stabilizer}) to  address transient errors.

Among the different models of distributed systems, specific models received more attention than others; and therefore the availability and versatility of building blocks differ from one model to another. For example, the synchronous no-failures model can automatically be extended in many different directions:\ asynchronous no-failures, synchronous self-stabilizing, asynchronous self-stabilizing, etc. Zooming-in to the world of self-stabilizing, there are various model-convertors: between shared-memory and message-passing, from an id-based to uniform system, etc. (see \cite{DolevSSBook}).

However, when moving away from the commonly researched models, the availability of such model-converters diminishes. In the current paper we are interested in an asynchronous network with \byzantine nodes and transient failures. That is, we aim at solving a problem in a way that is \byzantine tolerant, self-stabilizing and operates in an asynchronous network. The \byzantine adversary is assumed to be full-information and dynamic (sometimes called adaptive). There are few previous works that operate in similar models \cite{831130,DBLP:conf/opodis/SakuraiOM04,DBLP:conf/sss/MasuzawaT06}.
In these works non-faulty neighbors of \byzantine nodes may reach undesired states.
However, as far as we know, this is the first work operating in such a setting in which non-\byzantine
nodes reach their desired state even if they have \byzantine neighbors.

The problem we solve in the current work is clock-synchronization. Our solution assumes two simplifying  assumptions: a) a ``centralized daemon", \ie each node can run the entire algorithm as an atomic step; b) an ``en masse scheduler'' that adheres to the following: if $p$ gets scheduled twice, then $n-2f$ other non-faulty nodes get scheduled in between (formally defined in \definitionref{defenmass}). Under these assumptions we define and solve the clock synchronization problem in an asynchronous network while tolerating both \byzantine and transient faults. The solution is a randomized
algorithm with expected convergence time of $O(3^{n-2f})$.

Both assumptions can be seen as ``building blocks that do not yet exist''. When constructing a self-stabilizing asynchronous algorithm (without \byzantine nodes), it is reasonable to assume a centralized daemon due to the mutual exclusion algorithm of Dijkstra (see \cite{DijkstraSS}). Thus, once an equivalent algorithm can be devised for this paper's model, the first assumption can be removed.
In \sectionref{sec:enmasse} we provide an algorithm that implements the second assumption, thus allowing its usage without reducing the generality of a solution that uses it.

Due to our dependence on the first assumption, we consider this work as a step towards a full solution of the clock synchronization in an asynchronous network that is self-stabilizing and \byzantine tolerant. We hope it leads to further research of this model, one which will produce an equivalent of Dijkstra's algorithm operating in the current work's model.

\subsubsection*{Related Work}
Being able to introduce consistent ``time'' in a distributed system is an
important task, however difficult it may be in some models. For
many distributed tasks the crux of the problem is to synchronize
the operations of the different nodes. One method of doing this is
using some sort of ``time-awareness'' at each node, ensuring that
different ``clocks'' advance in a relatively synchronized manner. Therefore,
it is interesting to devise such algorithms that are highly robust.

In the past, various models were considered. Ranging from
synchronous systems (see \cite{PulseSync-DISC07,DolWelSSBYZCS04,DigiClock-SSS06}), in which all nodes receive a common signal
simultaneously at regular intervals; through bounded-delay
systems (see \cite{PulseSync-SSS07,1759107}), in which only a bound on the message delivery time is
given; to completely asynchronous systems (see
\cite{Unison,HermanPhaseClocks}), in which only an
eventual (but not bounded) delivery of messages is assumed.
Independent of the timing model, different fault tolerance
assumptions are considered: the self-stabilizing fault paradigm,
in which all nodes follow their protocol but may start with
arbitrary values of their variables and program counter (see
\cite{DolevSSBook}). Another commonly assumed faults are the
Fail-stop faults, in which some of the nodes may crash and cease
to participate in the protocol. Lastly, \byzantine faults are
considered to be the most severe fault model, as they assume that
the faulty nodes can behave arbitrarily and even collude in trying
to keep the system from reaching its designated goals (see
\cite{983102,LynchBook}).

``Knowing what time it is'' acquires different flavors in
different models. In systems without any faults, it is usually
assumed that each node has a physical clock, and these clocks
differ from node to node. The main issue is to synchronize the
different clocks as close as possible. In a synchronous,
self-stabilizing and \byzantine tolerant model, this problem was
termed ``digital clock synchronization'', and consisted of having
all nodes agree on some bounded integer and increase it every
round (see \cite{ADG91,PulseSync-DISC07,DolWelSSBYZCS04,DigiClock-SSS06}).

The traditional concept of ``clock synchronization'' does not hold in an asynchronous environment.
Therefore,
previous work has defined ``phase clocks'' or ``unison'' (see
\cite{Unison,HermanPhaseClocks}), which states that
each node has an integer valued clock, and neighboring nodes
should be at most $\pm 1$ from each other. It is shown (for example,
see \cite{HermanPhaseClocks}) how such ``synchronization'' is
sufficient in solving many problems.

Most previous works in the asynchronous model considered
self-stabilizing or \byzantine faults, but not both. In the
current work, we consider both  fault models. However, defining what ``telling the time''
means in an asynchronous, self-stabilizing and \byzantine tolerant
manner is a bit tricky. To address that, a new notion of ``knowing what
time it is'' is introduced: {\em a clock function}.

All previous clock-synchronization (or
phase-clock, or unison, etc.) algorithms can be viewed in the following way:
each time a non-faulty node is running, it executes some piece of code
(``function'') that returns a value (``the clock value'') and
there are constraints on the range of different non-faulty nodes' values.
In the synchronous digital clock synchronization problem, the function returns
an integer value, and we require that all nodes executing the
function at the same round receive exactly the same value and a node executing the
function in consecutive rounds receives consecutive values. In an
asynchronous network (\ie in \cite{HermanPhaseClocks}), different
nodes may execute their clock-functions at different times and at
different rates. The constraint on the returned values can be
described as follows: given a configuration of the system, if $p$
would execute its clock-function and receive a value $v$, then any
neighbor of $p$ that would execute its clock-function at the
same configuration, would receive a value that differs by at most $\pm 1$ from $v$.

In the current work it is assumed that the network is fully
connected, which means that every node is connected to every other
node. Therefore, the constraint of the clock-function is
simplified, informally requiring any two non-faulty nodes that
execute the clock-function to receive values that are at most one
apart.

In synchronous networks the problem of self-stabilizing \byzantine tolerant clock synchronization is equivalent to the problem of \byzantine agreement, in the sense that any solution to the self-stabilizing \byzantine tolerant clock synchronization problem is also a solution to the (non-self-stabilizing) \byzantine agreement problem.
In asynchronous networks the best known full-information dynamic \byzantine agreement has
expected exponential convergence time (see \cite{806707}). While the synchronous equivalence between clock synchronization and \byzantine agreement does not transfer to the asynchronous setting (as it strongly uses the fact that all nodes agree on the exact same clock value), it raises the possibility that improving the result of this paper will
require usage of new techniques. That is,
it is not known yet if the self-stabilizing clock synchronization of the current work can be used to solve \byzantine agreement. However, if it can be used, then improving the exponential convergence of the current work would lead to an improvement of the best known asynchronous \byzantine agreement against a dynamic full-information adversary.

\subsubsection*{Contribution}
Our contribution is three-fold. First, we define the clock-function problem, which is a
generalization of the clock synchronization problem. This definition provides a meaningful
extension of the clock synchronization problem to the asynchronous self-stabilizing
\byzantine tolerant model.

Second, we provide an algorithm that solves the clock-function
problem in the above model. Using shared memory, it has an expected $O(3^{n-2f})$
convergence time,
independent of the wraparound value of the clock. Notice that for synchronous networks,
the first two \mbox{contributions} were already presented in \cite{DolWelSSBYZCS04}. Our contribution is
with respect to asynchronous networks.

Lastly, in \sectionref{sec:enmasse} we construct a building block that bounds the relative
rates at which different non-faulty nodes progress with respect to other non-faulty nodes.
More specifically, between any two atomic steps of a non-faulty node $p$, there are
guaranteed to be atomic steps of $n-2f$ other non-faulty nodes. (See the ``en masse scheduler'' assumption
described in the introduction).
We postulate that this building block can be used in other asynchronous self-stabilizing \byzantine tolerant
settings.

\subsubsection*{Overview}
We start by defining the model (see \sectionref{sec:model}). A subset of all possible
runs is defined and denoted \emph{``en masse''} (see \definitionref{defenmass}).
\sectionref{sec:problemdef} discusses
different aspects of defining clock synchronization in an asynchronous, self-stabilizing, \byzantine tolerant
environment; and defines a clock function, which is a generalization of the clock synchronization problem.

\sectionref{sec:proposed_algorithm} introduces \clockAlg, an algorithm that solves the problem
at hand. \sectionref{sec:proofMain} contains the correctness proof for \clockAlg.
Both \sectionref{sec:proposed_algorithm} and \sectionref{sec:proofMain} are correct only for
en masse runs, for which \clockAlg requires fault redundancy of $n > 6f$.

In \sectionref{sec:enmasse} the algorithm \enmasse is presented,
which transforms any run into an
en masse run. Leading to the correctness of \clockAlg for any run. However, the transformation
done by \enmasse increases the fault redundancy of \clockAlg to $n > 12f$.
Lastly, \sectionref{sec:desc} concludes with a discussion of the results.

\section{Distributed Model}\label{sec:model}
The system is composed of a set of $n$ nodes denoted by $\p$. Every
pair of nodes $p,q \in \p$ communicates via shared memory (\ie a fully-connected communication graph), in an
asynchronous manner. That is, $p$ and $q$ share two registers:
$R_{p,q}, R_{q,p}$.\footnote{Pair-wise communication is used to allow \byzantine
nodes to present different values to different nodes; as opposed to assuming a
single register per-node that can be read by all other nodes.}
Register $R_{p,q}$ is written by $p$ and read
by $q$.\footnote{For simpler presentation we assume that $p$
writes and reads $R_{p,p}$.} A configuration $\C$ describes the
global state of the system and consists of the states of each
node and the state of each register. A run of the system is an
infinite sequence of configurations $\C_0 \rightarrow \C_{1}
\rightarrow \cdots \rightarrow \C_r \rightarrow \cdots$, such that
the configuration $\C_{r+1}$ is reachable from configuration
$\C_r$ by a single node's atomic step. In the context of the
current paper, an atomic step consists of reading all registers, performing
some calculation and then writing to all registers.

The system is assumed to start from an arbitrary
initial configuration $\C_0$. We show that eventually - in the
presence of continuous \byzantine behavior - the system becomes
synchronized.

In addition to transient faults, up to $f$ of the nodes may be
\byzantine. The \byzantine adversary has full
information, \ie it can read the values in every node's memory\footnote{Actually, the presented algorithm stores all its state in the shared registers.}
and in the shared registers between any two nodes. There are no private
channels and the adversary is computationally unbounded.
Moreover, the adversary
is dynamic, which means it may choose to ``capture'' a non-faulty
node at any stage of the algorithm. However, once the adversary
has ``captured'' $f$ nodes in some run, it cannot affect other
nodes and in a sense becomes static. The results of this paper can
be extended to the setting in which the adversary continues to be
dynamic throughout the run, as long as the adversary is limited by
the rate at which it can release and capture
non-faulty nodes. We do not present this extension in the current
paper for the sake of clarity. However, one can easily be
convinced that it applies, once the main points of the work are
explained.

The adversary also has full control of the scheduling of atomic steps and can use its
full information knowledge in this scheduling. However, for a clock synchronization algorithm
to be meaningful, runs in which some of the non-faulty nodes never get
to perform atomic steps should be excluded. Thus, throughout the paper only fair runs are considered:
\begin{definition}
  A run is {\bf fair} if every non-faulty node performs infinitely
  many atomic steps.
\end{definition}

A subset of all fair runs is defined:
\begin{definition}\label{defenmass}
A run $\mathcal{T}$ is {\bf en masse with respect} to node $p$ if for any 2 atomic steps $p$ performs during $\mathcal{T}$ (say at configurations $\mathcal{C}$ and $\mathcal{C}'$, respectively) there are at least $n-2f$ non-faulty nodes that perform atomic steps between $\mathcal{C}$ and $\mathcal{C}'$.

A run $\mathcal{T}$ is {\bf en masse} if it is fair and it is en masse with respect to all non-faulty nodes.
\end{definition}
As stated in the overview, en masse runs are needed for \clockAlg to operate correctly.
Assuming all runs are en masse runs, the fault tolerance redundancy required is $n > 6 f$\hideForShortVersion{ (see \lemmaref{lemmaref:twovaluesaretight} for an example of the necessity of $n > 6 f$)}. However,
in \sectionref{sec:enmasse} we show how to remove the requirement of en masse runs, at
the cost of increasing the fault tolerance redundancy to $n > 12 f$.

\section{Problem Definition}\label{sec:problemdef}
Before formally defining the problem at hand, consider the properties
a distributed clock synchronization algorithm should have in an asynchronous setting:
\begin{enumerate}
  \item {\it (clock-value)} a means of locally computing the current clock value at any non-faulty node;
  \item {\it (agreement)} if different non-faulty nodes compute the clock value close (in time) to each other, they should obtain similar values;
  \item {\it (liveness)} if non-faulty nodes continuously recompute clock values, then they should obtain increasing values.
\end{enumerate}
For example, in a synchronous network, the clock synchronization problem is usually formulated as: {\it (clock-value)} each node $p$ has a bounded integer counter $Clock_p$; {\it (agreement)} for any two non-faulty nodes $p,q$ it holds that $Clock_p = Clock_q$; {\it (liveness)} if $Clock_p=z$ at round $r$ then $Clock_p=z+1$ at round $r+1$. Since the clock is bounded, the previous sentence is slightly modified: ``if $Clock_p=z$ at round $r$, then $Clock_p=z+1 (\!\!\!\!\mod k)$ at round $r+1$''; where $k$ represents the wrap-around value of the
clock.

\begin{notation}
  Denote by $\modk{a}{b}$ the value $(a + b \mod k)$.
\end{notation}

In an asynchronous setting, it is impossible to ensure that all nodes update their clocks simultaneously. Thus, the ``agreement'' property requires a relaxed version as opposed to the synchronous setting's stricter version. In addition, the ``liveness'' property is somewhat tricky to define, due to the \byzantine presence. To illustrate the difficulty, consider a set of $f$ \byzantine nodes that ``behave as if'' they were non-faulty, and they repeatedly recompute the clock value. According to the definition above, the clock value will increase continuously, even though non-faulty nodes did not perform a single step. Therefore, such a clock synchronization algorithm is useless, as the \byzantine nodes can make it reach any clock value; in other words, the \byzantine nodes ``control'' the clock value.

It is not immediately clear how these ``benign'' \byzantine nodes can be differentiated from the non-faulty nodes. The following definitions address such difficulties, and present a formalization of the clock synchronization problem in this paper's model.

\hide{In the following
we assume $k$ to be large enough (\ie $k \geq 100$)\footnote{The
actual number required by the proofs is 6 or 7, but this point is
of little interest.}. In \sectionref{sec:fixk} we show how to remove
this assumption.\ezra{say something about $k$ being the wraparound}
}

\begin{definition}
  A value $v'$ is at most $d$ {\bf ahead of} $v$ if there exists
  $j,$ $0 \leq j \leq d$, such that $\modk{v}{j} = v'$.
    Denote ``$v'$ is at most $d$ ahead of $v$'' by $v \preceq_d v'$.
\end{definition}

\definitionref{def:clock-value} addresses the ``clock-value'' property:
\begin{definition}\label{def:clock-value}
  A {\bf clock-function} $\F$ is an algorithm that when executed during
  an atomic step returns a value in the range $\{0,...,k-1\}$.
  Denote by $\F_p(\C)$ the value returned when $p$ executes $\F$
  during an atomic step at configuration $\C$.
  \end{definition}

Consider the ``agreement'' property: it requires that different non-faulty nodes that compute the clock value simultaneously, receive similar values. What does ``simultaneously'' mean in an asynchronous setting? It can be captured by requirements on the clock values computed in different runs. In addition, the interference caused by \byzantine nodes in different runs needs to be captured.

Informally, ``agreement'' requires the following from $\F$: given a configuration $\C$, no matter what the
adversary does, if different non-faulty nodes execute $\F$ they receive values that are close to each other.
 \definitionref{def:adversarial-move}, \definitionref{def:welldefined} and \definitionref{def:runwelldefined} formally state the ``agreement'' requirement.
First, ``no matter what the adversary does'' is formally  defined:
\begin{definition}\label{def:adversarial-move}
  An {\bf adversarial move} from a configuration $\C$ is any
  configuration reachable by an arbitrary sequence of atomic steps of faulty nodes
  only.
\end{definition}

Second, ``different non-faulty nodes execute $\F$'' is divided into two cases. Let $p,q$ be non-faulty nodes. The first case considers the computed value of $p$ (when calculating $\F$ on $\C$) as opposed to the computed value of $q$ (see \definitionref{def:welldefined}). The second case considers the computed value of $q$ after $p$ has computed its value (see \definitionref{def:runwelldefined}). Both cases require the computed values of $p$ and $q$ to be close to each other.
\begin{definition} \label{def:welldefined}
  A configuration $\C$ is $\ell$-{\bf well-defined} (with respect to some clock-function $\F$) if there is a value $v$ s.t.
  for any non-faulty node $p$ and every adversarial move $\C'$ from $\C$ it holds that $v \preceq_{\ell} \F_p(\C')$. $v$ is called ``a defined value'' at $\C$. (There may be more than one such $v$).
  \end{definition}

Informally, \definitionref{def:welldefined} says that $\C$ is
$\ell$-well-defined if there is an intrinsic value $v$ such that any adversarial move cannot increase the clock-value by more than $\ell$. Thus, any two non-faulty nodes $p, q$ (in different
run extensions from $\C$) that execute $\F$ on $\C$ (no matter what the adversary
has done) will receive values in the range $\{v,\dots,v+\ell\}$; \ie $p$ and $q$'s values are at most $\ell$ apart.

Suppose $\C_0$ is $\ell$-well-defined with value $v$, and that a
non-faulty node $p$ performs an atomic step at $\C_0$ resulting in $\C_1$
and then a non-faulty node $q$ performs an atomic step at $\C_1$.
\definitionref{def:welldefined} does not imply any
constraint on the value of $\F_p(\C_0)$ with respect to
$\F_q(\C_1)$, therefore the following definition is required:

\begin{definition}\label{def:runwelldefined}
  A run is $\ell$-{\bf well-defined} (w.r.t. a clock-function
  $\F$) if: a) every configuration $\C$ in the run is
  $\ell$-well-defined; b) for two consecutive configurations $\C,
  \C'$, if $v$ is a defined value of $\C$ and $v'$ is a defined value of $\C'$ then $v \preceq_\ell v'$.
\end{definition}

\definitionref{def:runwelldefined} states that the values of a clock-function
$\F$ on consecutive configurations cannot be arbitrary. That is,
they must be at most $\ell$ apart from the previous configuration.
However, there is no requirement that they actually increase; \ie ``liveness'' is not captured by the previous definitions.

\begin{definition}\label{def:clock-sync}
  A run is $\ell$-{\bf clock-synchronized} (w.r.t. some clock-function $\F$), if
  it is $\ell$-well-defined (w.r.t. $\F$)
  and the defined values of consecutive configurations change
  infinitely many times. (\Ie for infinitely many consecutive configurations $\C, \C'$ the defined values of $\C$  differ from the defined values of $\C'$).
\end{definition}

Notice that \definitionref{def:runwelldefined} already requires that defined values of consecutive configurations are non-decreasing (assuming that $\ell$ is sufficiently small with respect to $k$). Thus, combined with \definitionref{def:clock-sync}, it implies that in an $\ell$-clock-synchronized run, infinitely many
configurations are configurations with increasing defined values (informally, ``increasing'' means that one defined value is achieved by adding less than $\frac{k}{2}$ to a previous defined value).

\begin{remark}
  \definitionref{def:runwelldefined} and \definitionref{def:clock-sync} impose requirements on the \emph{defined} values of consecutive configurations. However, a specific node $p$ might compute a clock value that is decreasing between consecutive configuration. \ie $p$'s clock might ``go backward''. For example, let $\C, \C'$ be two consecutive configurations, and let the defined value of
  both configurations be $v$. It is possible that $p$ will compute the clock value to be $v+1$ for $\C$, while computing the clock value to be $v$ for $\C'$.

  However, this possibility is immanent to an asynchronous
  \byzantine tolerant clock synchronization that has a wraparound value $k$. Consider a setting in which all nodes but one are advanced in a synchronous manner, while a single node $p$ performs atomic steps only once every $k-1$ rounds. In such a setting, $p$ should update its clock value to be slightly below its previous value (alternatively, it can be seen as increasing the value by $k-1$).
\end{remark}

\begin{definition}
  An algorithm $\A$ solves the $\ell$-{\bf clock-synchronization problem} if
  there is a clock-function $\F$ s.t.
  any fair run starting from any arbitrary configuration has a suffix that is $\ell$-clock-synchronized with respect to $\F$.
\end{definition}

An ideal protocol would solve the $0$-clock-synchronization
problem. However, due to the asynchronous nature of the discussed model, the best that can be expected is to solve the $1$-clock-synchronization problem.
We aim at solving the $\ell$-clock-synchronization problem for as many values of
$\ell \geq 1$ as possible. Clearly, if $\A$ solves the
$\ell_1$-clock-synchronization problem, then $\A$ also solves the
$\ell_2$-clock-synchronization problem for any $\frac{k}{2}
> \ell_2 \geq \ell_1$.

Therefore, the rest the paper concentrates on solving the $5$-clock-synchronization problem; thus, solving the $\ell$-clock-synchronization problem for all $\frac{k}{2} > \ell \geq 5$ . In
\sectionref{sec:to1clock} we show how to use any $\frac{k-1}{2}$-clock-synchronization problem to solve the $1$-clock-synchronization problem, thus solving the
$\ell$-clock-synchronization problem for all $\frac{k}{2} > \ell \geq 1$.

\begin{figure*}[t!]\center
\begin{minipage}{4.3in}
\hrule \hrule \vspace{1.7mm} \footnotesize
\setlength{\baselineskip}{3.9mm} \noindent Algorithm \clockAlg \hfill\textit{/* executed on
node $q$
*/}
 \vspace{1mm} \hrule \hrule
\vspace{1mm}

\linenumber{01:} {\bf do} forever:\\
\\
\makebox[0.93cm]{} \textit{/* read all registers */}\\
\linenumber{02:} \tb {\bf for} $i := 1$ to $n$  \\
\linenumber{03:} \due {\bf set} $val_i:=$ {\bf read} $R_{p_i,q} \mod k$; \\
\\
\makebox[0.93cm]{} \textit{/* some internal definitions */}\\
\linenumber{04:} \tb {\bf let} $\#v$ denote the number of times $v$ appears in $\{val_i\}_{i=1}^n$;  \\
\linenumber{05:} \tb {\bf let} $\countv(v,l)$ denote $\sum_{j=0}^l\#(\modk{v}{j})$;  \\
\linenumber{06:} \tb {\bf let} $\pass(l,a)$ denote $\{v | \countv(v,l) \geq a\}$;\\
\\
\makebox[0.93cm]{} \textit{/* update \myval */}\\
\linenumber{07:} \tb {\bf if} $\pass(0,n-f) \neq \emptyset$ then  \\
\linenumber{08:} \due {\bf set} $\myval_q := \modk{1}{\max \{\pass(0,n-f)\}}$;  \\
\linenumber{09:} \tb {\bf else if} $\pass(1,n-f) \neq \emptyset$ then  \\
\linenumber{10:} \due {\bf set} $\myval_q := \modk{1}{\max \{\pass(1,n-f)\}}$;  \\
\linenumber{11:} \tb {\bf else if} $\pass(1,n-2f) \neq \emptyset$ then  \\
\linenumber{12:} \due {\bf let} $low \notin \pass(1,n-2f)$ be such that $\modk{low}{1} \in \pass(1,n-2f)$;  \\
\linenumber{13:} \due {\bf let} $relative\_median = min \{l  | \l \geq 0\,\&\,\countv(low,l) > \frac{n}{2} \}$;  \\
\linenumber{14:} \due {\bf set} $\myval_q := \modk{low}{relative\_median}$;  \\
\linenumber{15:} \tb {\bf else} {\bf set} $\myval_q := $ randomly select a value from $\pass(1,n-3f) \bigcup \{0\}$;\\
\\
\makebox[0.93cm]{} \textit{/* write \myval to registers */}\\
\linenumber{16:} \tb {\bf for} $i := 1$ to $n$  \\
\linenumber{17:} \due {\bf write} $\myval_q$ into $R_{q,p_i}$; \\
\linenumber{18:} {\bf od};

\normalsize \vspace{1mm} \hrule\hrule
\end{minipage}
\begin{minipage}{5.2in}
 \caption{A self-stabilizing \byzantine tolerant algorithm  solving the $5$-Clock-Synchronization problem.}\label{figure:syncAlg}
 \end{minipage}
\end{figure*}

\section{Solving the 5-Clock-Synchronization Problem}\label{sec:proposed_algorithm}
An atomic step consists of reading all registers,
performing some calculations and writing to all registers. Thus,
an atomic step consists of executing once an entire ``loop'' of
\clockAlg (see \figureref{figure:syncAlg}).

Each non-faulty node $p$ has a bounded integer variable,
$\myval_p$, which represents the current clock value of $p$.
When $p$ performs an
atomic step, it reads all of its registers, thus getting an
impression of the clock values of the other nodes. It then
computes its own new clock value (which is saved in
$\myval_p$) and writes $\myval_p$ to all registers.

\clockAlg operates in a similar fashion to many other \byzantine tolerant algorithms.
It first gathers information regarding the clock value of the other nodes in the system.
Then it uses various thresholds to decide on the clock value for the next step.
If no threshold works (\ie no clear majority is found), it chooses a random value from
a small set of options.

To ensure all values read during \lineref{02-03} are in the range $\left[0, \dots, k-1\right]$, the algorithm
applies ``$\mod k$'' to the values read. This is a standard way of dealing with uninitialized values.

The crux of \clockAlg is in the exact thresholds and their application (Lines \linenumber{07}-\linenumber{15}).
In these lines, node $p$ considers different possibilities. Either
it sees a decisive majority towards some clock value (\lineref{07}
and \lineref{09}) in which case $p$ updates its local clock value
to coincide with the majority clock value it has seen.
Alternatively, if no clear majority exists (\lineref{15}),
 $p$ randomly selects a
new clock value. The interesting case is when $p$ sees a
``partial'' majority (\lineref{11}), in which case $p$ takes the relative median
of the clock values it has seen. We call this a ``relative
median'' since the clock values are ``$\!\!\!\mod k$'' and thus the
median in not well defined.

The full \clockAlg algorithm appears in
\figureref{figure:syncAlg}. \clockAlg solves the $\ell$-Clock-Synchronization problem for $\ell=5$; combined with the discussion at the end of \sectionref{sec:problemdef}, it shows how to solve the $\ell$-Clock-Synchronization problem for any $\frac{k}{2} > \ell \geq~5$.

\hideForShortVersion {
\begin{notation}
The value of any variable $var$ at configuration $\C_r$  is denoted by
$\C_r(var)$. For a node $q$ that does not perform the atomic step that changes
$\C_r$ to $\C_{r+1}$,  the value of $\myval_q$, denoted $\C_r(\myval_q)$, is the same before and after the atomic step.

For node $p$ that performs the atomic step at configuration $\C_r$, $\myval$ is the only variable that is not deterministically
determined by the values of the registers at the beginning of $p$'s atomic
 step. In such a case, for $\myval$, the notation $\C_r(\myval_p)$ denotes
  the value of $\myval$ before $p$ starts its atomic step, and
  $\C_{r+1}(\myval_p)$ denotes the value of $\myval$ after $p$
  finishes its atomic step.
  For all other variables,
  $\C_r(var)$ will denote the value of variable
 $var$, as computed for configuration $\C_r$;
  \ie $\C_r(\pass_p(1, n-2f))$ denotes the value
  that $p$ computes for $\pass(1, n-2f)$ during its atomic step
  at $\C_r$.
\end{notation}

} 

\section{Correctness Proof}\label{sec:proofMain}
In the following discussion we
consider the system only after all transient faults ended and each non-faulty node has taken
at least one atomic step.
We consider only runs of the system that begin after that initial sequence of atomic steps.

Informally, a round is a portion of a run such that each
node that is non-faulty throughout the round performs an atomic
step at least once.
The first round (of a run $\T$) is the
minimal prefix $\R$ of the run $\T$ such that each node that is
non-faulty throughout $\R$ performs an atomic step at least once.
Consider the suffix $\T'$ of $\T$ after the first round was
removed. The second round of $\T$ is the first round of $\T'$;
the definition continues so recursively.

Consider any fair run of the system $\C_0 \rightarrow \C_{1}
\rightarrow \cdots \rightarrow \C_r \rightarrow \cdots$,
and consider the transition from configuration $\C_r$ to configuration
$\C_{r+1}$, due to some (possibly faulty) node $p$'s atomic step. Since we
consider  only runs after each non-faulty node $q$ has taken
at least one atomic step past the end of the transient faults events, the value of
$\myval_q$ reflects the latest value written to all of $q$'s
write-registers. This property is true for all configurations that we consider. Thus,
regarding a non-faulty $p$ that performs an atomic step,  for all non-faulty $q$ it holds that $R_{q,p} =
\myval_q$.

\hideForLongVersion{
Due to space limitations, only an overview of the proof is given. The full proof can be found at
\cite{FullAsyncClock}.
} 

The proof outline is as follows. First, define a tight configuration:

\hideForLongVersion {
\begin{definition}
$H(\C_r, v, d)$ is the set containing every non-faulty node $q,$ such that $\myval_q$ (in~$\C_r$) is at most $d$ ahead of $v$.

  A configuration $\C_r$ is {\bf tight} around value $v$ if $|H(\C_r, v, 1)| \geq
  n-2f$; a configuration is tight if it is tight around some
  value.
\end{definition}
} 

Second, we show that if a configuration $\C_r$ is tight then so is $\C_{r+1}$.
Third, if $\C_r$ is not tight, then we show that with probability $\frac{1}{3^{n-2f}}$
some configuration within 2 rounds from $\C_r$ will be tight. Concluding that after an expected $O(3^{n-2f})$
rounds the system reaches a tight configuration; and all following configurations are tight as well.
At this stage, we need to show that the value $v$ that a configuration is tight around continuously increases.

To do so, we show that given that all configurations are tight, different non-faulty nodes that perform atomic steps can have values from a set containing (at most) 3 consecutive values. Moreover, for consecutive configurations, the minimal
value among these 3 values can increase by at most 3. Lastly, by closely analyzing the behavior of \clockAlg, we conclude that within 4 rounds the minimal value above increases. That is, the clock function value changes, and  changes again within at most 4 rounds, \ie the clock value changes infinitely many times.

The reason behind the increase of the aforementioned minimal value lies in the following claim: one of two things
can happen, either the minimal value increases, or all the non-faulty nodes' clock values become at most 1 apart. In the second scenario, after one round, the minimal value will increase. Concluding that
the clock value changes infinitely many times, as required.

\begin{remark}
The en masse property is used in the proof that if $\C_r$ is not tight, then with probability $\frac{1}{3^{n-2f}}$ a configuration within 2 rounds from $\C_r$ will be tight. Since in an en masse run some set of $n-2f$ different non-faulty nodes are required to take atomic steps in a consecutive manner. Together with a claim stating that
each such step has probability of $\frac{1}{3}$ to flip a coin ``in the right direction'', we get that with probability $\frac{1}{3^{n-2f}}$ a tight configuration is reached.
\end{remark}

\hideForShortVersion{

\begin{lemma}\label{lemma:aheadof}
  If $v_1 \preceq_d v'$ and $v_2 \preceq_d v'$ then either $v_2 \preceq_d v_1$ or $v_1 \preceq_d
  v_2$.
\end{lemma}
\begin{proof}
  By definition, there are $j_1,j_2$ ($0 \leq j_1,j_2 \leq d$) such that
  $v'=\modk{v_1}{j_1}$ and $v'=\modk{v_2}{j_2}$. Thus,
  $\modk{v_2}{j_2}=\modk{v_1}{j_1}$, which means that $v_2 =
  \modk{v_1}{(j_1-j_2)}$. Clearly, $|j_1-j_2| \leq d$. If $j_1-j_2 \geq 0$ then
  $v_2$ is at most $j_1-j_2 \leq d$ ahead of $v_1$. Otherwise, $j_2-j_1 >
  0$, meaning that $v_1 = \modk{v_2}{(j_2-j_1)}$. That is, $v_1$ is
  at most $j_2-j_1 \leq d$ ahead of $v_2$.
  \end{proof}

We are interested in the set of non-faulty nodes that are
``close'' to each other with respect to their value of $\myval$.
\begin{definition}
  $H(\C_r, v, d)$ is the set containing any non-faulty node $q,$ such that $\C_r(\myval_q)$ is at most $d$ ahead of $v$. Formally, $H(\C_r, v, d) = \{\mbox{non-faulty $q$}  ~|~ v \preceq_d \C_r(\myval_q) \}$.
\end{definition}

\begin{definition}
  $H(\C_r, p, v, d)$ is the set containing any node $q$, such that $\C_r(R_{q,p})$ is at most $d$ ahead of $v$. Formally, $H(\C_r, p, v, d) = \{q ~|~ v \preceq_d \C_r(R_{q,p}) \}$.
\end{definition}

Notice that $H(\C_r, v, d)$ contains only non-faulty nodes, while
$H(\C_r, p, v, d)$ may contain faulty nodes. The difference stems
from $H(\C_r, p, v, d)$ representing what $p$ ``perceives'' at
configuration $\C_r$, as opposed to $H(\C_r, v, d)$ which says
``what is true'' in configuration $\C_r$.

\begin{lemma}\label{lemma:HvsH}
  $|H(\C_r, v, d)| \geq |H(\C_r, p, v, d)|-f$ and $|H(\C_r, p, v, d)| \geq |H(\C_r, v,d)|$.
\end{lemma}

\begin{remark}\label{remark:Hvspass}
  Notice that $H(\C_r, p, v, d)$ contains all the nodes (including
  faulty nodes) whose registers' value (in $\C_r$) is at most $d$
  ahead of $v$. Thus, $v \in C_r(\pass_p(d, x)) \neq
  \emptyset$ if (and only if) $|H(\C_r, p, v, d)| \geq x$.
\end{remark}


\begin{definition}
  A configuration $\C_r$ is {\bf tight} around value $v$ if $|H(\C_r, v, 1)| \geq
  n-2f$; a configuration is tight if it is tight around some
  value.
\end{definition}

\begin{lemma}\label{lemmaref:twovaluesaretight}
  If a configuration $\C_r$ is tight around value $v$ and around
  value $v' \neq v$, then either $v \preceq_1 v'$, or $v' \preceq_1 v$.
\end{lemma}
\begin{proof}
  By the lemma's assumption, it holds that $|H(\C_r, v, 1)| \geq
  n-2f$ and $|H(\C_r, v', 1)| \geq n-2f$. Since $n > 6f$, there is
  some non-faulty node $q \in H(\C_r, v, 1) \bigcap H(\C_r, v',
  1)$. Thus, $\C_r(\myval_q)$ is at most 1 ahead of $v$ and at
  most 1 ahead of $v'$. The rest follows from
  \lemmaref{lemma:aheadof}.
\end{proof}

\begin{remark}\label{remark:nomorethan2values}
  Following the same line of proof as in \lemmaref{lemmaref:twovaluesaretight}  shows that $\C_r(\pass_p(1, n-2f))$ can contain at
  most 2 values, and these values are consecutive values.
\end{remark}

\begin{lemma}\label{lemma:lowwelldefined}
  If in configuration $\C_r$, non-faulty node $p$ performs
  \lineref{12} then $\C_r(low_p)$ is well defined, for $k \geq 3$.
\end{lemma}
\begin{proof}
  By \remarkref{remark:nomorethan2values}, if $p$ passes the
  condition of \lineref{11} then $\C_r(\pass_p(1, n-2f))$ contains
  at most 2 values, which are consecutive. Thus, if $k \geq 3$ then
  $\C_r(low_p)$ is well defined.
\end{proof}

\begin{lemma}\label{lemma:middleandlowarev}
  If $p$ passes the condition in \lineref{11} and $|H(\C_r, v, 1)| \geq n-3f$ then
  $v \preceq_1 (\modk{low}{relative\_median})$, for $k > 4$.
\end{lemma}
\begin{proof}
  $p$ passed the condition in \lineref{11}, thus $\C_r(\pass_p(1,
  n-2f)) \neq \emptyset$. Thus, for some $v'$ it holds that
  $|H(\C_r, p, v',1)| \geq n-2f$ (see \remarkref{remark:Hvspass}), and therefore $|H(\C_r, v',1)| \geq n-3f$ (see \lemmaref{lemma:HvsH}).
  By the lemma's assumption, $|H(\C_r, v, 1)| \geq n-3f$. Since $n > 6f$, there is some non-faulty node $q \in H(\C_r, v',1)
  \bigcap H(\C_r, v, 1)$. That is $v' \preceq_1 \C_r(\myval_q)$
  and $v \preceq_1 \C_r(\myval_q)$. By \lemmaref{lemma:aheadof}
  either $v \preceq_1 v'$ or $v' \preceq_1 v$.

  According to \lineref{12} and \remarkref{remark:nomorethan2values}, $\modk{low}{1} \preceq_1
  v'$. Thus, in both scenarios ($v \preceq_1 v'$ or $v' \preceq_1
  v$) it holds that $low \preceq_3 v$. Informally, $low$ is ``before'' $v$, and
  $relative\_median$ (see \lineref{13}) is increased until there are more than $\frac{n}{2}$ nodes in the range $[low, \modk{low}{relative\_median}]$. Since there are
  $\geq n-3f > \frac{n}{2}$ copies of ``$v$'', $relative\_median$ will be such that
  $\modk{low}{relative\_median} \in \{ v, \modk{v}{1}\}$.

  Formally,
  $relative\_median$ is the minimal value such that
  $\countv_p(low, relative\_median)$ contains more than $\frac{n}{2}$ nodes.
  Since $|H(\C_r, v, 1)| \geq n-3f > \frac{n}{2}$, at least one copy of ``$v$''
  is counted towards the sum of $\countv_p(low, relative\_median)$.
  Since $k > 4$ (by the lemma's assumption), copies of $v$ will
  not be counted in $\countv_p(low, relative\_median')$ for $relative\_median'$ such
  that $\modk{low}{relative\_median'} \notin \{ v, \modk{v}{1}\}$.
  On the other
  hand, $\countv_p(low, relative\_median') \geq n-2f$ for $relative\_median'$ such
  that $\modk{low}{relative\_median'} = \modk{v}{1}$.
  Thus, $\modk{relative\_median}{low}=v$ or $\modk{relative\_median}{low}=\modk{v}{1}$. In both cases $v \preceq_1 \modk{low}{relative\_median}$.
\end{proof}

\begin{lemma}\label{lemma:tightclosure}
  If a configuration $\C_r$ is tight then so is $\C_{r+1}$.
\end{lemma}
\begin{proof}
  If $p$ is faulty, its update of $\myval_p$ and/or its
  write-registers do not affect the ``tightness'' of the
  configuration $\C_{r+1}$. Thus, the rest of the proof assumes
  that $p$ is non-faulty.

  First, notice that if $\C_r(\pass_p(0, n-f)) \neq \emptyset$
  then $\C_r(\pass_p(0, n-f))$ contains a single value. This is
  because $\C_r(\pass_p(0, n-f))$ contains all the values $v$
  that appear at least $n-f$ times in the registers read by $p$.
  If two values appear more than $n-f$ times they must be the same
  value (since $n > 6f$).

  Consider $p$ updating $\myval_p$. If $p$ updates it in
  \lineref{08}, then it must have passed the ``if'' in \lineref{07}.
  Thus, $\C_r(\pass_p(0, n-f)) \neq \emptyset$, which means
  that $\C_r(\pass_p(0, n-f)) = \{v\}$. Thus,  $\myval_p$ is
  updated to $\modk{v}{1}$. From \remarkref{remark:Hvspass} it holds that $|H(\C_r, p, v, 0)| \geq n-f$,
  and from \lemmaref{lemma:HvsH} it holds that $|H(\C_r, v, 0)| \geq
  n-2f$. Thus, after $p$'s update of $\myval_p$ to $\modk{v}{1}$ (and $p$'s
  writing $\myval_p$ to all of $p$'s write-registers), $|H(\C_{r+1}, v, 1)| \geq
  n-2f$ holds. Thus, configuration $\C_{r+1}$ is tight.

  If $p$ updates $\myval_p$ in \lineref{10}, then $\C_r(\pass_p(1, n-f)) \neq
  \emptyset$. Denote by $v := \max \{\pass_p(1, n-f)\}$.
  (in \lineref{10}, $\myval_p$ is updated to $\modk{v}{1}$).
  Notice
  that $v \preceq_1 \myval_p$. In addition, $v \in \C_r(\pass_p(1,
  n-f))$, thus (by \remarkref{remark:Hvspass}), $|H(\C_r, p, v, 1)| \geq
  n-f$. Therefore, after $p$'s update of $\myval_p$, $p \in H(\C_{r+1}, v,
  1),$ which means that $|H(\C_{r+1}, v, 1)|\geq |H(\C_r, v, 1)|$
  (because $p$ may not be counted for in $H(\C_r, v,
  1)$). By \lemmaref{lemma:HvsH}, $|H(\C_r, v, 1)| \geq n-2f$,
  thus we have that $|H(\C_{r+1}, v, 1)| \geq n-2f$, \ie configuration $\C_{r+1}$ is
  tight.

  We are left to consider updates of $\myval_p$ in \lineref{14} and \lineref{15}. Notice that since $\C_r$ is
  tight, $|H(\C_r, p, v, 1)| \geq n-2f$, for some $v$. Thus,
  $\C_r(\pass_p(1,n-2f)) \neq \emptyset$. Therefore, $p$ passes the condition of
  \lineref{11}, and $p$ does not perform \lineref{15}.

    According to the lemma's assumption $\C_r$ is tight around some
  value $v$, and by \lemmaref{lemma:middleandlowarev} we have
  that $v \preceq_1 \modk{relative\_median}{low}$. That is, $p$ updates
  $\myval_p$ such that $p \in H(\C_{r+1}, v ,1)$. Since, $|H(\C_r, v
  ,1)| \geq n-2f$ and since $p$ is the only node changing its
  $\myval_p$ at $\C_r,$  it holds that $|H(\C_{r+1}, v ,1)| \geq n-2f$.
  That is, $\C_{r+1}$ is tight.
%
\end{proof}

The following lemmas assume all runs are en masse. In \sectionref{sec:enmasse} this assumption is removed (at the cost of reducing the fault tolerance to $n > 12f$).

\begin{lemma}\label{lemma:n2fnodes}
  Consider an en masse run, and a non-faulty node $p$ performing an atomic step at configuration $C_{r}$. Denote by $S_i$ the set of non-faulty nodes that have performed an atomic step between $C_{r}$ and $C_{r+i}$.  Let $m$ be the minimal $m$ s.t. $|S_m|= n-2f$, then each node $q \in S_m$ performed an atomic step exactly once between $C_{r}$ and $C_{r+m}$.
\end{lemma}

\begin{proof}
  First, since we consider only fair runs, eventually $S_{i}$ will contain all non-faulty nodes. Thus, $m$, as defined in the lemma is well defined. Assume by way of contradiction that some node $q$ in $S_m$ performed two atomic steps between $C_r$ and $C_{r+m}$. In that case, there must be $n-2f$ non-faulty nodes that perform atomic steps between $q$'s two atomic steps. Thus, all of these nodes must be in $S_m$, leading to the fact that $|S_{m-1}|\geq n-2f$ which contradicts $m$ being the minimal $m$ s.t. $|S_m|= n-2f$. Therefore, there is no such $q \in S_m$, and all nodes in $S_m$ perform a single atomic step between $C_r$ and $C_{r+m}$.
\end{proof}

\begin{lemma}\label{lemma:evntuallconverge}
  Let $\C_{r_1}$ denote the first configuration of some round $\R$, and
  $\C_{r_2}$ denote the last configuration of round $\R+1$. With probability at least $\frac{1}{3^{n-2f}}$ configuration $\C_{r_2}$ is tight.
\end{lemma}

\begin{proof}
  Consider non-faulty nodes performing atomic steps between $C_{r_1}$ and $C_{r_2}$. If some configuration $\C_r, r_1 \leq r \leq r_2$ is
  tight, then - by using \lemmaref{lemma:tightclosure} - every
  configuration after $\C_r$ is tight. Thus, $\C_{r_2}$ is tight.
  Therefore, our target is proving that with probability at least $\frac{1}{3^{n-2f}}$ some configuration $\C_r$ is tight.

  Let $\C_r, r_1 \leq r \leq r_2$ be some configuration, and let
  $p$ be a non-faulty node performing an atomic step on $\C_r$.
  $p$ performs exactly one of the following:
  \lineref{08}, \lineref{10}, \lineref{14} or \lineref{15}.
  If $p$ performs \lineref{08} or \lineref{10}, then $\C_r(\pass_p(1, n-f)) \neq
  \emptyset$. That is, for some value $v \in \C_r(\pass_p(1, n-f))$ it
  holds that $|H(\C_r, p, v, 1)| \geq n-f$ which means that $|H(\C_r, v, 1)| \geq
  n-2f$ (see \lemmaref{lemma:HvsH}); that is, $\C_r$ is tight around $v$. Therefore, if any non-faulty node performs \lineref{08} or
  \lineref{10} during rounds $\R,\R\!+\!1$, then $\C_{r_2}$ is tight.

  The rest of the proof assumes no non-faulty node performs either \lineref{08} or
  \lineref{10} on any configuration $C_r, r_1 \leq r \leq r_2$. Consider the first $n-2f$ non-faulty nodes performing atomic steps in round $\R$. (By \lemmaref{lemma:n2fnodes} these nodes perform exactly one atomic step, \ie the adversary cannot reschedule a node if ``it does not like'' the outcome of that node's random coin). If they all perform only \lineref{15}, then there is some probability that they all choose the same value of \myval, as they all choose from a set that contains ``0''. Each node chooses from a set $\pass(1,n-3f) \bigcup \{0\}$, which contains at most 3 items.\ezra{explain why at most 3 items} Thus, with probability at least $\frac{1}{3^{n-2f}}$ all non-faulty nodes choose the same value, leading to a tight configuration.

  The proof continues under the assumption that some non-faulty node
  $p$ performs \lineref{14} on some configuration $\C_r$ during round $\R$.
  Using the notations $S_m$ of \lemmaref{lemma:n2fnodes}, there are $n-2f$ non-faulty nodes that perform atomic steps between $C_{r}$ and $C_{r'} = C_{r+m}$. Notice that since all non-faulty nodes perform an atomic step during round $\R+1$ then configuration $C_{r'}$ is reached in round $\R$ or in round $\R+1$, and in any case $C_{r'}$ is reached before $C_{r_2}$.

  Since $p$ performs \lineref{14} it passed the condition of \lineref{11} and it holds that
  $\C_r(\pass_p(1,n-2f)) \neq \emptyset$. By \lemmaref{lemma:HvsH}
  and \remarkref{remark:Hvspass} $|H(\C_r,v,1)| \geq n-3f$ for some
  value $v$. According to \lemmaref{lemma:middleandlowarev}, $v
  \preceq_1 \C_{r+1}(\myval_p)$; thus, $|H(\C_{r+1},v,1)| \geq
  n-3f$.

  The proof continues by showing that if $|H(\C_{r''},v,1)| \geq
  n-3f$ for $r'', r+1 \leq r'' \leq r'$, then with probability at least $\frac{1}{3}$ it holds that $|H(\C_{r''+1},v,1)| \geq n-3f$; and if node $q$ performing an atomic step on $\C_{r''}$ is non-faulty then also $q \in
  H(\C_{r''+1},v,1)$. Assume that $|H(\C_{r''},v,1)| \geq
  n-3f$ and consider a node $q$ performing an atomic step on $\C_{r''}$. If $q$ is faulty then its action does not change the value of $H(\C_{r''+1},v,1)$, thus $|H(\C_{r''+1},v,1)| \geq
  n-3f$. If $q$ is non-faulty and it performs \lineref{15} then since
  $v \in \C_{r''}(\pass_q(n-3f, 1))$, with probability at least $\frac{1}{3}$, $q$ selects $v$ as its value of $\myval_q$ (as it is selected from a set containing at most three items).
  On the other hand, if $q$ performs \lineref{14}, then since $|H(\C_{r''},v,1)| \geq
  n-3f$ by \lemmaref{lemma:middleandlowarev} $q \in H(\C_{r''+1},v,1)$ and $|H(\C_{r''+1},v,1)| \geq
  n-3f$. Thus, in either case with probability at least $\frac{1}{3}$ it holds that $q \in
  H(\C_{r''+1},v,1)$ and $|H(\C_{r''+1},v,1)| \geq n-3f$.

  Therefore, if at some configuration $\C_{r}$ it holds that $|H(\C_{r},v,1)| \geq
  n-3f$, then any non-faulty node $q$ operating in a configuration
  $\C_{r''}$, $r'' \geq r$ will have $\C_{r''+1}(\myval_q) \in
  H(\C_{r''+1},v,1)$ (with probability $\geq \frac{1}{3}$). Therefore, once $n-2f$ non-faulty nodes perform an
  atomic step, they are all in $H(\C_{r'},v,1)$ with probability at least $\frac{1}{3^{n-2f}}$.
  Thus, $\C_{r'}$ is tight with probability $\geq \frac{1}{3^{n-2f}}$.
\end{proof}

From this point on, the discussion assumes that all configurations
are tight. Therefore, \lineref{15} will never be executed.

\begin{definition}
  Denote by $\V(\C_r)$ the set containing any value $v$ of a non-faulty node $p$, such that $p$ ``helps''
  in the configuration $\C_r$ being tight around $v$.
  Formally, $$\V(\C_r) = \bigcup_{v \ s.t. \ |H(\C_r, v, 1)| \geq n-2f} \{\C_r(\myval_p) | p \in H(\C_r, v, 1) \}\;.$$
\end{definition}

\begin{lemma}\label{lemma:1of3}
 If $k \geq 6$, then  $\V(\C_r)$ is exactly one of the following: $\{v\}$,
  $\{v, \modk{v}{1}\}$ or \\$\{v, \modk{v}{1}, \modk{v}{2}\}$, for some value $v$.
\end{lemma}

\begin{proof}
From \lemmaref{lemmaref:twovaluesaretight} it follows that
$|\V(\C_r)| \leq 3$. Moreover, if $k \geq 6$ then for any two
values $v,v' \in \V(\C_r)$ it holds that $v \preceq_2 v'$ or $v'
\preceq_2 v$, which is proved by way of contradiction. Assume that neither hold;
notice that $v \in \V(\C_r)$ due to some value $\overline{v}$ such
that $\overline{v} \preceq_1 v$ and $|H(\C_r, \overline{v}, 1)|
\geq n-2f$; for similar reasons $v' \in \V(\C_r)$ due to
$\overline{v}' \preceq_1 v'$. Thus, if neither $v \preceq_2 v'$ or
$v' \preceq_2 v$, we have that $H(\C_r, \overline{v}, 1) \bigcap
H(\C_r, \overline{v}', 1) = \emptyset$, leading to $|H(\C_r,
\overline{v}, 1) \bigcup H(\C_r, \overline{v}', 1)| \geq 2(n-2f) =
2n-4f > n$. From the above discussion, if $|\V(\C_r)| = 3$, it
must be of the form $\V(\C_r) =\{v, \modk{v}{1}, \modk{v}{2}\}$.

If $|\V(\C_r)| = 2$, then $\V(\C_r)$ can either be
$\{v,\modk{v}{1}\}$ or $\{v,\modk{v}{2}\}$. In the second option,
no non-faulty node has a value of $\modk{v}{1}$, that is, $H(\C_r,
v, 1) \bigcap H(\C_r, \modk{v}{2}, 1) = \emptyset$. As before, we
reach a contradiction from $|H(\C_r, v, 1) \bigcup H(\C_r,
\modk{v}{2}, 1)| \geq 2(n-2f) = 2n-4f > n$.

Therefore, $\V(\C_r)$ is exactly one of the following: $\{v\}$,
$\{v, \modk{v}{1}\}$ or $\{v, \modk{v}{1}, \modk{v}{2}\}$.
\end{proof}

\lemmaref{lemma:1of3} leads to defining the ``minimal'' and
``maximal'' values of $\V(\C_r)$ in the following way:
$\V_{min}(\C_r) :=\{v| v \in \V(\C_r) \,\&\, \modk{v}{-1} \notin \V(\C_r)\}$
and $\V_{max}(\C_r) :=\{v| v \in \V(\C_r) \,\&\, \modk{v}{1} \notin
\V(\C_r)\}$. By the above lemma both $\V_{min}(\C_r)$ and
$\V_{max}(\C_r)$ are well defined (for $k \geq 6$).

\begin{lemma}\label{lemma:consecutive}
  Let $k > 6$ and let $\C_r, \C_{r+1}$ be two consecutive configurations. Then, $\V_{min}(\C_r) \preceq_3 \V_{min}(\C_{r+1})$.\ezra{Lemma 10. why $k>6$ and $=<_3$ (perhaps explain to the reader).}
\end{lemma}

\begin{proof}
  Let $p$ be the node that performs an atomic step between $\C_r$
  and $\C_{r+1}$. If $p$ is faulty, then its update of $\myval_p$
  does not affect the value of $\V_{min}(\C_{r+1})$, and we have that
  $\V_{min}(\C_r)=\V_{min}(\C_{r+1})$, which means that
  $\V_{min}(\C_r) \preceq_3 \V_{min}(\C_{r+1})$.

  The rest of the proof assumes $p$ is non-faulty. $p$ updates
  $\myval_p$ due to \lineref{08}, \lineref{10} or \lineref{14}. If
  $p$ performs \lineref{08} then $\myval_p$ is updated to $\modk{v}{1}$, where $v=\max
  \{\pass_p(0, n-f)\}$. By definition, $v \in \V(\C_r)$; thus,
  $\V_{max}(\C_r) \preceq_1 \V_{max}(\C_{r+1})$.

  Similarly, if $p$ performs \lineref{10}, then $\myval_p$ is updated to $\modk{v}{1}$, where $v=\max
  \{\pass_p(1, n-f)\}$. Again, by definition $v \in \V(\C_r)$; thus,
  $\V_{max}(\C_r) \preceq_1 \V_{max}(\C_{r+1})$.

  Consider $p$ performs \lineref{14}. By definition, $|H(\C_r, \V_{min}(\C_r), 1)| \geq
  n-3f$; thus, by \lemmaref{lemma:middleandlowarev} $\V_{min}(\C_r) \preceq_1 (\modk{low}{relative\_median})$. Thus,
  $\V_{max}(\C_r) \preceq_1 \V_{max}(\C_{r+1})$.

  In all 3 scenarios it was shown that $\V_{max}(\C_r) \preceq_1
  \V_{max}(\C_{r+1})$. However, $|\V(\C_r)| \leq 3$ (see
  \lemmaref{lemma:1of3}) and thus $\V_{min}(\C_{r}) \preceq_2
  \V_{max}(\C_{r})$. Therefore, $\V_{min}(\C_{r}) \preceq_3
  \V_{max}(\C_{r+1})$. Since $\V_{min}(\C_{r+1}) \preceq_2
  \V_{max}(\C_{r+1})$, we have that $\V_{min}(\C_{r}) \preceq_3
  \V_{min}(\C_{r+1})$; as required.
\end{proof}

\begin{lemma} \label{lemma:nextvalinV}
  If a non-faulty node $p$ performs an atomic step between $\C_r, \C_{r+1}$ then
  $\C_{r+1}(\myval_p) \in \V(\C_{r+1})$.
\end{lemma}

\begin{proof}
  If $p$ performs \lineref{08} or \lineref{10} then
  $\C_{r+1}(\myval_p)= \modk{v}{1}$ for some $v \in
  \C_r(\pass_p(1, n-f))$. By \remarkref{remark:Hvspass} and \lemmaref{lemma:HvsH},
  $|H(\C_r, v, 1)| \geq n-2f$. Therefore, $p \in H(\C_{r+1}, v, 1)$ and $|H(\C_{r+1}, v, 1)| \geq
  n-2f$, which means that $\C_{r+1}(\myval_p) \in
  \V(\C_{r+1})$.

  Consider $p$ performing \lineref{14}. Since $\C_r$ is
  tight there is some value $v$ such that $|H(\C_r, v, 1)| \geq
  n-2f$. By \lemmaref{lemma:middleandlowarev} we have that $v \preceq_1 \C_{r+1}(\myval_p)$. Thus,
  $p \in H(\C_{r+1}, v, 1)$ and  $|H(\C_{r+1}, v, 1)| \geq
  n-2f$, which means that $\C_{r+1}(\myval_p) \in
  \V(\C_{r+1})$.
\end{proof}

\begin{lemma}\label{lemma:livlyness}
  Starting from a tight  configuration $\C'$, within 4 rounds there are two consecutive configurations $\C_r, \C_{r+1}$ for which $\V_{min}(\C_r) \neq \V_{min}(\C_{r+1})$.
\end{lemma}

\begin{proof}
  Consider configurations $\C'=\C_{r_1}, \C_{r_2}, \C_{r_3}, \C_{r_4}, \C_{r_5}$ such that
  $\C_{r_{i+1}}$ is one round after $\C_{r_i}$. Let $\C_{r'}, r_1 \leq r' < r_5$ be some configuration,
  and let $p$ be a non-faulty node performing an atomic step on
  $\C_{r'}$.
  If $\V_{min}(\C_{r'}) \neq \V_{min}(\C_{r'+1})$, we are
  done. Otherwise, assume by way of contradiction that for all $r_1 \leq r' < r_5$
  it holds that $\V_{min}(\C_{r'}) =
  \V_{min}(\C_{r'+1})$; denote $\V = \V_{min}(\C_{r_1})$. Therefore, by \lemmaref{lemma:1of3} and
  \lemmaref{lemma:nextvalinV}, $\C_{r'+1}(\myval_p) \in H(\C_{r'+1}, \V,
  2)$. Since all non-faulty nodes have performed an atomic step between $\C_{r_1}$ and $\C_{r_2}$, for any $\C_{r'},
  r_2 \leq r' < r_5$ it holds that $|H(\C_{r'}, \V, 2)| = n-f$. We continue to consider only
  configurations $\C_{r'}$ such that $r_2 \leq r' < r_5$.

  Notice that if $p$ performs \lineref{08} or \lineref{10}
  then $\myval_p$ is updated to be at least ``$+1$'' from $\V$.
  This is because only $\V,\modk{\V}{1},\modk{\V}{2}$ may be in
  $\C_{r'}(\pass_p(1, n-f))$ or $\C_{r'}(\pass_p(0, n-f))$. Thus, taking the maximum
  of these sets and adding ``1'' produces a value that is at least ``$+1$'' from $\V$.

  We divide the proof into two scenarios:
  1) for some configuration $\C_{r'}, r_2 \leq r' \leq r_4$ it holds that
  $|H(\C_{r'},
  \V, 0)| < \frac{n}{2}-f$; 2) for all $\C_{r'}, r_2 \leq r' \leq r_4$ it holds that $|H(\C_{r'}, \V, 0)| \geq
  \frac{n}{2}-f$. Consider the first case, ant let $\C_{r'}$ be some configuration s.t. $|H(\C_{r'},
  \V, 0)| < \frac{n}{2}-f$. Clearly, if $p$
  performs \lineref{08} or \lineref{10} then it updates $\myval_p$
  to be ``greater'' than $\V$. If $p$ performs \lineref{14}, then
  because values that are not $\modk{\V}{1},\modk{\V}{2}$ can appear at most $\frac{n}{2}$ times, $p$ must
  update $\myval_p$ to be ``greater'' than $\V$. Therefore, if $|H(\C_{r'},
  \V, 0)| < \frac{n}{2}-f$, then also $|H(\C_{r'+1},
  \V, 0)| < \frac{n}{2}-f$. Moreover,
  $\C_{r'+1}(\myval_p) \in \{\modk{\V}{1},\modk{\V}{2}\}$.

  Thus, if for some configuration $\C_{r'}, r_2 \leq r' \leq r_4$ it holds that $|H(\C_{r'},
  \V, 0)| < \frac{n}{2}-f$, then starting from $\C_{r_5}$ (at least one round after $\C_{r'}$),
  no non-faulty node has $\myval$ equal to $\V$,
  which means that $\V_{min}(\C_{r_5}) \neq
  \V=\V_{min}(\C_{r_1}$).

  We continue under the assumption that $|H(\C_{r'}, \V, 0)| \geq
  \frac{n}{2}-f$, for all $r_2 \leq r' \leq r_4$. Recall that all non-faulty nodes have values
  from the set $\{\V, \modk{\V}{1}, \modk{\V}{2}\}$. Thus,
  $|H(\C_{r'}, \modk{\V}{1}, 1)| \leq \frac{n}{2}$. Therefore, if
  $p$ passes \lineref{08} or \lineref{10} then $\C_{r'}(\pass_p(0,n-f))$
  and
  $\C_{r'}(\pass_p(1,n-f))$ do not contain $\modk{\V}{1}, \modk{\V}{2}$; which means they may contain $\modk{\V}{-1}$ or $\V$.
  Thus, $p$ updates $\myval_p$ to $\V$ or $\modk{\V}{1}$. On the other hand, if $p$
  performs \lineref{14} it updates $\myval_p$ to be either $\V$ or
  $\modk{\V}{1}$ (recall that $\V_{min}(\C_{r'}) = \V$ which means that $|H(\C_{r'}, \V, 1)| \geq n-2f > \frac{n}{2}$). Thus, in all cases, $p$ updates $\myval_p$ to be
  either $\V$ or $\modk{\V}{1}$. Therefore, $|H(\C_{r_3}, \V, 1)| = n-f$, and $|H(\C_{r'}, \V, 1)| = n-f$ for all $r_3 \leq r' \leq r_4$.

  Thus, any
  non-faulty node performing an atomic step on configuration $C_{r'}, r_3 \leq r' \leq r_4$, either passes the
  condition of \lineref{08} or \lineref{10}. In both cases, it
  updates $\myval_p$ to be ``greater'' than $\V$. Thus, starting from $\C_{r_4}$
  it holds that $|H(\C_{r'}, \V, 0)| <
  \frac{n}{2}-f$. And we are back to the previous case, in which
  we have shown that $\V_{min}$ must change within 1 round. Thus,
  for some configuration $\C_{r'}, r_4 \leq r' < r_5$ it holds
  that $\V_{min}(\C_{r'}) \neq \V$. In other words, within 4 rounds
  there is some configuration $\C_r$ such that $\V_{min}(\C_r)
  \neq \V_{min}(\C_{r+1})$.
\end{proof}
} 

Following is the main result of the paper, which is shown to be true assuming that the runs are en masse. In the following section en masse runs are constructed from fair runs. Thus, the theorem can be updated to only require that the run is fair.

\begin{theorem}\label{thm:5clocksync}
  \clockAlg solves the $5$-clock-synchronization problem within expected $O(3^{n-2f})$ rounds, for any en masse run and wrap-around value greater than 6 (\ie $k > 6$).
\end{theorem}
\hideForLongVersion{
} 
\hideForShortVersion{
\begin{proof}
Define the clock-function $\F$ executed by non-faulty node $p$
at configuration $\C_r$ to be the value of $\C_{r+1}(\myval_p)$
as updated by \clockAlg when executed as an atomic step. Combining
\lemmaref{lemma:1of3}, \lemmaref{lemma:consecutive} and
\lemmaref{lemma:nextvalinV} shows that a tight configuration is
$5$-well-defined with respect to $\F$; where $\V_{min}(\C_r)$ a defined value at $\C_r$. In addition, these lemmas
show that any fair run $\T$ consisting of only tight
configurations is $5$-well-defined. By \lemmaref{lemma:livlyness},
run $\T$ is also $5$-clock-synchronized.

Given any en masse run $\T'$ and any initial configuration $\C_0$,
\lemmaref{lemma:evntuallconverge} states that with probability $\geq \frac{1}{3^{n-2f}}$ there is some configuration $\C_r \in \T'$ (within two rounds from $\C_0$)
that is tight. By \lemmaref{lemma:tightclosure} every
configuration after $\C_r$ is also tight. Thus, every fair run $\T'$ has a suffix $\T$ that consists of only tight configurations; and this suffix is reached within $O(3^{n-2f})$ rounds in expectation. From the above paragraph, $\T$ is $5$-clock-synchronized.

Thus, \clockAlg solves the $5$-clock-synchronization problem.
\end{proof}
} 

\hideForLongVersion {
\begin{remark}
  The requirement that $k > 6$ stems from the analysis of the relative median and the different
  updates performed in \clockAlg. Due to lack of space, we do not go into details.
  Full details are available in \cite{FullAsyncClock}.
\end{remark}
} 

\section{Ensuring En Masse Runs}\label{sec:enmasse}
Our goal is to ensure that if a non-faulty node $p$ performs a step, at least $n-2f$ non-faulty nodes have performed a step since $p$'s last step. That is, given an algorithm $\A$ we want to ensure that if some non-faulty node performs two steps of $\A$ then there are at least $n-2f$ different non-faulty nodes that also perform steps of $\A$. To ensure this, we present an algorithm \enmasse that ensures that a specific action, denoted ``act'', is executed twice by the same non-faulty node $p$ only if there are at least $n-4f$ other non-faulty nodes that have also executed ``act''. By setting ``act'' to execute an atomic step of $\A$, we achieve the required goal. \Ie \clockAlg will
 be executed entirely every time ``act'' appears in \enmasse.

 As the algorithm we present ensures only $n-4f$ nodes execute ``act'' in between two ``acts'' of every non-faulty node, we must reduce the \byzantine tolerance by half ($n > 12f$) to use \enmasse as a subcomponent of \clockAlg.
That is, \clockAlg requires a threshold of $\frac{2}{3}n$ non-faulty nodes ($n-2f$ threshold for $n > 6f$); \enmasse ensures a threshold of $n-4f$. Therefore, by reducing the fault tolerance to $n > 12f$ we ensure that $n-4f > \frac{2}{3}n$, as required by \clockAlg.

Our solution borrows many ideas from \cite{citeulike:3995616}. Due to our model's atomicity assumptions, each node can read all registers and write to all registers in a single atomic step. Thus, the problems that \cite{citeulike:3995616} encounters do not exist in the current paper at all. However, in the current model there are additional faults (\byzantine and self-stabilizing) which do not exists in \cite{citeulike:3995616}. Interestingly, the same ideas used in \cite{citeulike:3995616} can be adapted to the self-stabilizing \byzantine tolerant setting.

For each node $p$, there is a set of labels $Labels_p$  associated with $p$. In addition, each node $p$ has
a variable $label_p$ from the set $Labels_p$; Also, $p$ has an ordering vector $order_p$, of length $|Labels_p|$, which induces an order on the labels in $Labels_p$. Lastly, each node $p$ has a time-stamp $time_p$, which is a vector of $n$ entries, consisting of a single label $time_p[q] \in Labels_q$ for each node $q$.

\begin{definition}
  A label $b$ is of {\bf type} $p$ if $b \in Label_p$.
\end{definition}

\begin{definition}
  Two labels $b,c$ of type $p$ are compared according to $order_p$, where $b<_pc$ if $b$ appears before $c$ in the vector $order_p$. The inequalities $\leq_p$, $>_p,\geq_p,=_p$ are similarly defined.
\end{definition}

\begin{definition}
Given two time-stamps $time_p, time_q$, and a set of nodes $I$, we say that
$time_p >_I time_q$ if $p,q \in I$ and for every entry $i \in I$, $time_p[i] \geq_i time_q[i]$, $time_p[q]=_q time_q[q]$ and $time_p[p] >_p time_q[p]$.
\end{definition}

To simplify notations, when it is clear from the context, we write $p >_I q$ instead of $time_p >_I time_q$. That is, when comparing nodes (according to $>_I$), we actually compare the nodes' time stamps.
\begin{definition}
A set $I$ of nodes is {\bf comparable} if for any $p,q \in I$
either $p >_I q$ or $q >_I p$.
\end{definition}

\begin{lemma}\label{lemma:lemmacomp}
  If $I$ is a comparable set, and $p,q,w \in I$, and $p >_I q$, $q >_I w$ then $p >_I w$.
\end{lemma}
\hideForShortVersion{
\begin{proof}
  Since $I$ is comparable, either $p >_I w$ or $w >_I p$. Suppose by way of contradiction that $w >_I p$, thus $time_p[w] <_w time_w[w]$.
  However, since $p >_I q$ we have that $time_p[w] \geq_w time_q[w]$, and since $q >_I w$ we have that $time_q[w] =_w time_w[w]$.
  Thus, $time_p[w] \geq_w time_w[w]$, contradicting $time_p[w] <_w time_w[w]$. Therefore, it is not true that $w >_I p$, leaving only one other option: $p >_I w$.
\end{proof}
} 

\hideForShortVersion{
\begin{lemma}\label{lem:pqrelate}
Let $I, I'$ be comparable sets, then $I \cap I'$ is a comparable set. Moreover,
for $p,q \in I \cap I'$,
$p <_{I \cap I'} q$ iff $p <_I q$.
\end{lemma}
\ezra{for final version, check if this proof can be removed}
\begin{proof}
  Let $I, I'$ be comparable sets, and let $p, q \in I \cap I'$. Since $p,q \in I$ and $I$ is comparable, either $p<_I q$ or $q <_I p$, similarly, either $p <_{I'}q$ or $q <_{I'} p$. Suppose by way of contradiction that $p<_I q$ and $q <_{I'} p$. Due to $p<_I q$ it holds that $time_q[q] >_q time_p[q]$ and due to $q <_{I'} p$ it holds that $time_q[q] =_q time_p[q]$; leading to a contradiction. Thus, either $p<_I q$ and $p<_{I'} q$ or $p>_I q$ and $p>_{I'} q$.

  Assume that $p<_I q$ and $p<_{I'} q$. Therefore, for all $i \in I \cup I'$ it holds that
   $time_p[i] \leq_i time_q[i]$, $time_p[p] = time_q[p]$ and $time_p[q] < time_q[q]$.
    Thus, for all $i \in I \cap I'$ it holds that $time_p[i] \leq_i time_q[i]$, and by definition
     we have that $p<_{I \cap I'} q$. Similarly, if $p>_I q$ and $p>_{I'} q$ then $p >_{I \cap I'} q$.

  It was shown that only two options exist: 1) $p<_I q$ and $p<_{I'} q$, 2) $p>_I q$ and $p>_{I'} q$. If option 1 occurs, then $p<_{I \cap I'} q$; if option 2 occurs then $p>_{I \cap I'} q$. Thus, any $p,q \in I \cap I'$ either $p>_{I \cap I'} q$ or $p<_{I \cap I'} q$ holds. \ie $I \cap I'$ is comparable. Moreover, we have shown that $p <_{I \cap I'} q$ iff $p <_I q$, as required.
\end{proof}

\begin{remark}
Notice that proof of the lemma above also implies that $p <_I q$ iff $p <_{I'} q$.
\end{remark}

} 

Notice that a comparable set $I$
induces a total order among the elements in $I$, therefore we can refer to the index of an element in $I.$
\begin{definition}
  A node $p \in I$ is said to be the $k$th highest (in $I$) if $\left|\{q \in I | q >_I p\}\right|=k-1$. Let $I_\#(p)=k$ if $p \in I$ is the $k$th highest in $I$.
\end{definition}

The 1st highest in $I$ is the node that is larger than all other nodes. The 2nd highest node in $I$ is the node that has only one node larger than it; (and so on).

\hideForShortVersion{
Informally, we wish to show that given two intersecting comparable sets $I, I'$, if a node $p$ is the $i$th highest item in  $I$, it is at most $i+\ell$ highest in $I'$; where $\ell$ changes according to $I, I'$. The following lemma formally bounds the difference between $I_\#(p)$ and $I'_\#(p)$.

\begin{lemma}\label{lemma:differentIs}
Let $I, I'$ be comparable sets, and denote $\ell = |I| - |I \cap I'|$. If $p \in I \cap I'$ then $I_\#(p) \leq I'_\#(p) + \ell$.
\end{lemma}
\begin{proof}
Let $p \in I \cap I'$. By definition $I_\#(p)=\left|\{q \in I | q >_I p\}\right|+1$ and $I'_\#(p)=\left|\{q \in I' | q >_{I'} p\}\right|+~1$. Therefore, it is enough to show that $|\{q \in I | q >_I p\}| \leq |\{q \in I' | q >_{I'} p\}| + \ell$. Consider the set $A = \{q \in I \cap I' | q >_{I \cap I'} p\}$, be \lemmaref{lem:pqrelate}, it holds that $A \subseteq \{q \in I | q >_I p\}$ and $A \subseteq \{q \in I' | q >_{I'} p\}$. Clearly, $\{q \in I | q >_I p\} - A$ contains only items in $I- I \cap I'$. Thus, $|\{q \in I | q >_I p\} - A| \leq |I| - |I \cap I'|$ and since $A \subseteq \{q \in I | q >_I p\}$ it holds that $|\{q \in I | q >_I p\}| - |A| \leq |I| - |I \cap I'|$. That is,
$|\{q \in I | q >_I p\}| \leq |A| + \ell$. Since, $A \subseteq \{q \in I' | q >_{I'} p\}$ it holds that
$|A| \leq |\{q \in I' | q >_{I'} p\}|$. Thus, $|\{q \in I | q >_I p\}| \leq  |\{q \in I' | q >_{I'} p\}| + \ell$, as required.
\end{proof}

\begin{corollary}\label{cor:differentIs}
  Let $I, I'$ be comparable sets, and denote $\ell' = \max\{|I|,|I'|\} - |I \cap I'|$. If $p \in I \cap I'$ then $I'_\#(p) - \ell' \leq I_\#(p) \leq I'_\#(p) + \ell'$.
\end{corollary}
}

\subsection{Algorithm \enmasse}
This section proves general properties of comparable sets. It discusses ``static'' sets, that do not change over time. The following algorithm considers comparable sets that change from step to step. However, during each atomic step, the comparable sets that are considered do not change, and the claims from the previous section hold.
That is, when reasoning about the progress of the algorithm, the comparable sets that are considered are all ``static''.

In the following algorithm, instead of storing both $label_p$ and $time_p$, each node stores just $time_p$ and the value of $label_p$ is the entry $time_p[p]$. In addition, during each atomic step, the entire algorithm is executed, \ie a node reads all time stamps and all order vectors of other nodes, and can update its own time stamp during an atomic step.

\begin{figure*}[t!]\center

\begin{minipage}{4.8in}
\hrule \hrule \vspace{1.7mm} \footnotesize
\setlength{\baselineskip}{3.9mm} \noindent Algorithm \enmasse \hfill\textit{/* executed on
node $q$
*/}
 \vspace{1mm} \hrule \hrule
\vspace{1mm}

\linenumber{01:} {\bf do} forever:\\
\\
\makebox[0.93cm]{} \textit{/* read all registers and initialize structures */}\\
\linenumber{02:} \tb {\bf for each} node $p$, read $time_p$ and $order_p$; \\
\linenumber{03:} \tb {\bf set} $\I := \emptyset$; \\
\linenumber{04:} \tb {\bf for each} set $W \subseteq \p$ s.t. $|W| \geq n-f$ and $q \in W$:  \\
\linenumber{05:} \due {\bf construct} $I := \{time_p \ | \ p \in W\}$; \\
\linenumber{06:} \due {\bf if} $W$ is comparable then $\I := \I \cup \{I\}$; \\
\\
\makebox[0.93cm]{} \textit{/* decide whether to execute ``update'' and whether to execute ``act''~*/}\\
\linenumber{07:} \tb {\bf if} for some $I \in \I$, it holds that $I_\#(q) \geq n-3f$ then \\
\linenumber{08:} \due update $time_q, order_q$ and ``act'';\\
\linenumber{09:} \tb {\bf if} $\I = \emptyset$ then update $time_q, order_q$;\\
\linenumber{10:} \tb {\bf write} $time_q$ and $order_q$;\\
\linenumber{11:} {\bf od};
\vspace{1mm}
\hrule
\vspace{1mm}
Updating $time_q$ is done by setting $time_q[p]=label_p$, for every $p \in \p$.\\
Updating $order_q$ consists of changing the order induced by $order_q$ such that $label_q$ is first and for other labels the order is preserved.
\normalsize \vspace{1mm} \hrule\hrule
\end{minipage}

 \caption{A self-stabilizing \byzantine tolerant algorithm ensuring en masse runs.}\label{figure:enmasse}
\end{figure*}

When a node $q$ performs an update, it changes the value of $time_q$ and $order_q$ in the following way: a) $order_q$ is updated such that $time_q[q]$ is larger than any other label in $Labels_q$. b) $time_q[p]$ is set to be $time_p[p],$ for all $p$. Notice that the new $order_q$ does not affect the relative order of labels in $Labels_q$ that are not $time_q[q]$. That is, if $l_1, l_2 \neq time_q[q]$ and $l_1 \leq_q l_2$ before the change of $order_q$, it holds that $l_1 \leq_q l_2$ also after updating the  $order_q$.

Intuitively, the idea of \enmasse is to increase the time stamp of a node $q$ only if $q$ sees that most
of the other nodes are ahead of $q$. When the time stamp is increased, $q$ also performs ``act''. This leads
to the following dynamics: a) If $q$ has performed an ``act'' twice, \ie updated its time stamp twice, then
after the first update, $q$ is ahead of all other nodes. b) However, since $q$ is ahead of all non-faulty nodes, if $q$ updates its time stamp again
it must mean that many nodes have updated their time stamps after $q$'s first update. \ie between two ``act'' of $q$ many other nodes have performed ``act'' as well.

\hideForLongVersion{
In a similar manner to \sectionref{sec:proofMain}, the lemmas and proofs are available in the full version \cite{FullAsyncClock}.
} 

We continue with an overview of the proof. First, consider the set of non-faulty nodes, and
consider the set of time stamps of these nodes. The proof shows that if this set is comparable for
some configuration $\C_r$ then it is comparable for any configuration $\C_{r'}$ where $r' > r$.
Second, we consider an arbitrary starting state, and consider the set $Y_r$  containing non-faulty nodes that have updated their time stamp by the end of round $r$. It is shown that if $|Y_r| \geq n-2f$ then  $|Y_{r+1}| \geq n-f$. Moreover, if $|Y_r| < n-2f$ then $|Y_{r+1}| \geq |Y_r| +1$. Thus, we conclude that within $O(n)$ rounds all non-faulty nodes have performed an update.

Once all non-faulty nodes have performed an update since the starting state, it holds that the set of
all non-faulty nodes' time stamps is comparable. Thus, during every round at least $2f$ nodes perform an update (as
they see themselves in the lower $3f$ part of the comparable set). This ensures that within $\frac{n}{2f}$ rounds some node will perform ``act'' twice. That is, there is no deadlock in the \enmasse algorithm. To conclude the
proof, it is shown that when the set of all non-faulty nodes values is comparable and some non-faulty node
performs ``act'' twice, it must be that another $n-4f$ non-faulty nodes have performed ``act'' in between.

\hideForShortVersion{
\begin{definition}
Let $Z$ be a set of non-faulty nodes, and consider an atomic step on configuration $\C_r$. Denote by $\TS_{\C_r}(Z)$ the set of time-stamps of nodes in $Z$, as they are at the beginning of the atomic step. When the configuration $\C_r$ is clear from the context, we simply say ``$Z$ is comparable'', instead of ``$\TS_{\C_r}(Z)$ is comparable''.
\end{definition}

\begin{lemma}\label{lemma:stayscomparable}
Let $Z$ be a set of non-faulty nodes and consider any atomic step on configuration $\C_r$. If $\TS_{\C_r}(Z)$ is comparable, then $\TS_{\C_{r'}}(Z)$ is comparable for any $r' \geq r$.
\end{lemma}

\begin{proof}
First, notice that whether $\TS_{\C_{r'}}$ is comparable or not depends only on the values of nodes in $Z$ and is not affected by nodes not in $Z$ . Therefore, only changes incurred by nodes in $Z$ matter. Consider the first node $p \in Z$ to perform an update at some configuration $\C_{r''}$, $r''\ge r$. Thus, $p$ sets $time_p[q] = time_q[q]$ for all nodes $q \in Z$ and also $p$ ensures that $time_{p}[p] >_p time_q[p]$ for all nodes $q \in Z, q \neq p$. Thus, $p >_Z q$ for all nodes $q \in Z, q \neq p$.

  Consider two nodes $q_1, q_2 \neq p$ in $Z$. $p$'s update does not change the value of $time_{q_1}[q]$ and $time_{q_2}[q]$ for all $q \in Z, q \neq p$. What about the relative order of $time_{q_1}[p]$ and $time_{q_2}[p]$?   W.l.o.g. $q_1 <_{Z} q_2$ before $p$'s update.
  According to the way $p$ changes $order_{p}$, after $p$'s update
  $time_{q_1}[p] \leq_p time_{q_2}[p]$. Therefore, $q_1 <_Z q_2$.

  Thus, for any pair of nodes $q_1,q_2 \in Z$ either $q_1 <_Z q_2$ or $q_1 <_Z q_2$ after $p$'s update. Repeating the above line of proof for any node in $Z$ inductively proves that $\TS_{\C_{r'}}(Z)$ is comparable.
\end{proof}

Consider the system starts in an arbitrary state. Denote by $U_r$ the set of non-faulty nodes that have not performed any ``update'' by the end of round $r$, and by $Y_r$ the set of non-faulty nodes that have performed ``update'' by the end of round $r$.

\begin{lemma}\label{lemma:YrComparable}
  The set $Y_r$ is comparable during the last configuration of round $r$.
\end{lemma}

\begin{proof}
Consider the order at which non-faulty nodes performed updates by the end of round $r
$: let $p_1$ denote the first node to perform an update, $p_2$ the second, $\dots, p_m$
be the $m$th (and last) non-faulty node to perform an update. A node may appear more
than once in that order, for example, if some node was the 2nd and 5th to perform an
update, $p_2 = p_5$. Denote by $A_i$ the set containing all non-faulty nodes in $\{p_1,
\dots, p_i\}$, where $p_i$ is the $i$th node performing an update. Clearly, $A_m = Y_r$,
and it is left to
show that the set of time-stamps of nodes from $A_m$ is comparable. Notice that if $p_{i
+1} \in A_i$ then $A_i = A_{i+1}$, that is, if the node performing the ``next'' update has
already performed an update, $A_i=A_{i+1}$.

  We show something stronger: for all $0 \leq i \leq m$ let $\C_i$ be the configuration after $p_i$ performs an atomic step, then $\TS_{\C_i}(A_i)$ is comparable. The proof is by induction on $i$.
  For $i=0,1$, clearly $A_i$ is comparable. We are left to show that if $A_i$ is comparable so is $A_{i+1}$.

If $p_{i+1}\in A_i$ (\ie $p_{i+1}$ already performed an update)\ then by
\lemmaref{lemma:stayscomparable} it holds that $A_{i+1}$ is comparable.
Consider $p_{i+1}$ such that $p_{i+1} \notin A_i$. According to the way $p_{i+1}$ updates
its registers, for any two nodes $q_1, q_2$, if $q_1<_{A_i} q_2$ then also $q_1 <_{A_{i
+1}} q_2$. Moreover, for any node $q \in A_i$ it holds that $q <_{A_i} p_{i+1}$. Thus, for
any pair of nodes $p,q \in A_{i+1}$ either $p <_{A_{i+1}} q$ or $q <_{A_{i+1}} p$.

Thus, $A_{i+1}$ is comparable. To complete  the proof, recall that $A_m = Y_r$.
\end{proof}

\begin{lemma}\label{lemma:Kenterline2}
Let $q$ be a node in $U_r$. Whenever $q$ performs an atomic step before the end of round $r$, it holds that $\I \neq \emptyset$, and for all $I \in \I$ it holds that $I_\#(q) < n-3f$.
\end{lemma}

\begin{proof}
  If $\I = \emptyset$ during $q$'s atomic step, then $q$ will perform an updated, and $q \notin U_r$. Similarly, if for some $I \in \I$ it holds that $I_\#(q) \geq n-3f$ then $q$ will also perform an update. Since $q \in U_r$, by definition, $q$ does not perform an update until the end of round $r$.
\end{proof}

\begin{lemma}\label{lemma:InK2}
Let $q$ be a non-faulty node and consider $q$'s atomic step during round $r'', r < r'' \leq r'$. For any comparable set $I \in \I$\ that $q$ considers in \lineref{07}, and for any $p' \in U_{r'}, p'' \in Y_r$: if $p',p'' \in I$ then $I_\#(p'') < I_\#(p')$.
\end{lemma}

\begin{proof}
Since $p',p''$ are both in $I$, either $p' <_I p''$ or $p'' <_I p'$. Since $p' \in U_{r'}$ it has
not yet performed an update, while $p''$ has performed an update before the end of
round $r$. Thus, $time_{p''}[p'] = time_{p'}[p']$ and therefore it cannot be that $p'' <_I p'$,
leaving us with $p' <_I p''$. Since $I$ is totaly ordered, any node $w$ such that $p'' <_I w$
also holds that $p' <_I w$. Therefore, there are more nodes in $I$ that are larger than
$p'$ than there are nodes that are larger than $p''$. \ie $I_\#(p'') < I_\#(p')$.
\end{proof}


\begin{lemma} \label{lemma:eventuallAllII}
If $|Y_r| \geq n-2f$ then $|Y_{r+1}| \geq n-f$.
\end{lemma}

\begin{proof}
  If $U_{r+1} = \emptyset$ we are done. Otherwise, assume by way of contradiction that $q \in U_{r+1}$. By
  \lemmaref{lemma:Kenterline2} during $q$'s atomic steps in round $r+1$ it holds that $\I \neq \emptyset$ and for all $I \in \I$ we have that $I_\#(q) < n-3f$. Consider such an $I \in \I$; since $|Y_r| \geq n-2f$, it holds that $|I \cap Y_r| \geq n-3f$. That is, $I$ contains at least $n-3f$ nodes that have performed an update. Thus, by \lemmaref{lemma:InK2}, $I_\#(q) \geq n-3f$ which contradicts the fact that $I_\#(q) < n-3f$. Thus, $q \notin U_{r+1}$ and $U_{r+1} = \emptyset$.
\end{proof}

\begin{lemma} \label{lemma:eventuallAllI}
If $|Y_r| < n-2f$ then $|Y_{r+1}| \geq |Y_r| + 1$.
\end{lemma}

\begin{proof}
Assume by way of contradiction that $|Y_{r+1}| < |Y_r| + 1$. Therefore, $|Y_{r+1}| \leq |Y_r| < n-2f$ and
$|U_{r+1}| \geq f+1$. Thus, for any set $I$ containing at least $n-f$ nodes, it holds that $I \cap U_{r+1} \neq \emptyset$.

Let $q \in U_{r+1}$, by \lemmaref{lemma:Kenterline2} during $q$'s atomic steps in round $r+1$ it holds that $\I \neq \emptyset$ and for all $I \in \I$ we have that $I_\#(q) < n-3f$. By \lemmaref{lemma:InK2}, for any node $q' \in U_{r+1}$ and any node $p' \in Y_r$ it holds that if $q',p' \in I$ then $I_\#(p') < I_\#(q')$. Therefore, for any node $p \in I \cap U_{r+1}$ it holds that $I_\#(p) > |I \cap Y_r|$.
As there are $|I \cap U_{r+1}| > 0$ nodes from $U_{r+1}$ in $I$, there is a node $p \in I \cap U_{r+1}$ such that
$I_\#(p) \geq |I \cap Y_r| + |I \cap U_{r+1}|$.

Notice that $Y_{r} \subseteq Y_{r+1}$, and since $|Y_{r+1}| \leq |Y_r|$ it holds that $Y_r = Y_{r+1}$. Moreover, $U_{r+1} \cap Y_{r+1} = \emptyset$ and $|U_{r+1} \cup Y_{r+1}| = n-f$. Thus, $|I \cap Y_r| + |I \cap U_{r+1}| =
|I \cap (Y_r \cup U_{r+1})| \geq n-2f$. That is, there is a node $p \in I \cap U_{r+1}$ such that $I_\#(p) \geq n-2f$. \ie there are at most $3f-1$ nodes $p'$ that have the following property: $time_{p'}[p'] <_{p'} time_p[p']$.
Since $p$ does not perform an update, this property can change only if $p'$ performs an update, which will reduce the number of nodes such that $time_{p'}[p'] <_{p'} time_p[p']$. Therefore, throughout round $r+1$, if $p$ considers a comparable set $I \in \I$, it will always have $I_\#(p) \geq n-3f$.

Consider an atomic step by $p$ during round r+1, by \lemmaref{lemma:Kenterline2} for any $I' \in \I$ that $p$ considers in \lineref{07}, $I'_\#(p) < n-3f$, which contradicts the above fact that $I_\#(p) \geq n-3f$. Thus, we conclude that $|Y_{r+1}| > |Y_r|$; as required.
\end{proof}

\begin{corollary}\label{cor:emptyUr}
  For any round $r \geq n-2f+2$, it holds that $U_r = \emptyset$ and $|Y_r|=n-f$.
\end{corollary}

\begin{proof}
  Apply \lemmaref{lemma:eventuallAllI} during the first $n-2f$ rounds, then apply \lemmaref{lemma:eventuallAllII} for round $n-2f+1$.
\end{proof}

\begin{lemma}\label{lemma:final}
For any round $r > n-2f+2$, any non-faulty node $q$ considers $Y_r \in \I$ during atomic steps of round $r$.
\end{lemma}

\begin{proof}
  By \corollaryref{cor:emptyUr}, $|Y_{r-1}| = n-f$ and $U_{r-1} = \emptyset$, leading to $q \in Y_{r-1}$.
  By \lemmaref{lemma:YrComparable} and \lemmaref{lemma:stayscomparable} $Y_r$   is comparable, and is viewed as comparable by $q$ during any atomic step of round $r$.  Thus, when $q$ performs an atomic step during round $r$,
  $q$ adds $Y_r$ to $\I$ at \lineref{06}.
\end{proof}

\begin{corollary}\label{cor:noline9}
  For any round $r > n-2f+2$, no non-faulty node passes the condition of \lineref{09}.
\end{corollary}

\begin{proof}
  By \lemmaref{lemma:final}, a non-faulty node $q$ performing an atomic step during round $r > n-2f+2$ has $\I \neq \emptyset$. Thus, the condition of \lineref{09} does not hold.
\end{proof}

\begin{lemma}\label{lemma:2fupdate}
  Let $r$ be any round, $r > n-2f+2$. During round $r$ at least $2f$ non-faulty nodes perform update.
\end{lemma}

\begin{proof}
  Let $r > n-2f+2$ be any round. Consider the set $W=\{time_q | q \in Y_r\}$ before the first atomic step of round $r$.
  Denote by $Z = \{p \in W | W_\#(p) \geq n-3f\}$, that is, $Z$ contains all nodes that are at most $n-3f$ highest in $W$. Since $|W| = n-f$ it holds that $|Z| = 2f$.

  For each node $q \in Z$, consider $q$'s first atomic step in round $r$. By the proof of \lemmaref{lemma:YrComparable}, since $q$ did not perform an update since the beginning of round $r$, when it performs its first atomic step, it holds that $W_\#(q) \geq n-3f$ and $q$ will pass the condition in \lineref{07} and perform an update (and ``act'') in \lineref{08}. This holds for all $q \in Z$, that is, for at least $2f$ nodes.
\end{proof}

\begin{lemma}\label{lemma:updateeveryfround}
  Starting from round $n-2f+3$, every non-faulty node $p$ performs an update at least once every $\frac{n}{2f}$ rounds.
\end{lemma}

\begin{proof}
  By the proof of \lemmaref{lemma:2fupdate}, every round the lowest $2f$ nodes perform an update. Thus,
  if $p$ does not perform an update during round $r$, there are at least $2f$ nodes higher than it. Consider round $r+i$, if $p$ does not perform an update there are $2f \cdot i$ nodes higher than $p$. Therefore, after at most
  $\frac{n}{2f}$ rounds $p$ will perform an update.
\end{proof}

\begin{lemma}\label{lemma:inbetween}
  Consider a non-faulty node $p$ performing an update twice, then there are at least $n-4f$ other nodes that have performed update in between.
\end{lemma}

\begin{proof}
  Consider the comparable set $Y_r$ after $p$'s first update. By the way $p$ does an update, ${Y_r}_\#(p) = 1$. When $p$ performs its second update, it has some comparable set $I \in \I$, such that $I_\#(p) \geq n-3f$. Therefore, at least $n-4f$ non-faulty nodes have become larger than $p$. Thus, they all must have performed an update.
\end{proof}
} 

\begin{theorem}\label{thm:main22}
  Starting from round $n-2f+3$, between any non-faulty node's two consecutive ``act''s, there are $n-4f$ non-faulty nodes that perform ``act''. Moreover, every non-faulty node performs an ``act'' at least once every $\frac{n}{2f}$ rounds.
\end{theorem}
\hideForShortVersion{
\begin{proof}
  By \corollaryref{cor:noline9}, non-faulty nodes perform update only if they also perform an ``act''. By \lemmaref{lemma:inbetween}, between a non-faulty node's two consecutive updates there are $n-4f$ non-faulty nodes that perform an update. By \lemmaref{lemma:updateeveryfround} every non-faulty node performs an update at least once every $\frac{n}{2f}$ rounds. Combining these two claims yields the required result.
\end{proof}
} 

\theoremref{thm:main22} states that using \enmasse one can ensure that nodes executing \clockAlg will have the following properties: 1) every non-faulty node $p$ executes an atomic step of \clockAlg once every $\frac{n}{2f}$ rounds; 2) if non-faulty $p$ executes 2 atomic steps of \clockAlg, then at least $n-4f$ non-faulty nodes execute atomic steps of \clockAlg in between. By setting $n > 12f$, these properties ensure that a fair run $\T$ is an en-masse run $\T'$ with respect to \clockAlg, s.t. each round of $\T'$ consists of at most $\frac{n}{2f}$ rounds of $\T$.

\section{Discussion}\label{sec:desc}
\subsection{Solving the 1-Clock-Synchronization Problem}\label{sec:to1clock}
First, the 5-clock-synchronization problem was solved using \clockAlg while
assuming en masse runs. Second, the assumption of en masse runs was removed in
\sectionref{sec:enmasse}. In this subsection we complete the paper's
result by showing how to transform a 5-clock-synchronization algorithm to a 1-clock-synchronization algorithm.

Given
any algorithm $\A$ that solves the $\ell$-clock-synchronization
problem, one can construct an algorithm
$\A'$ that solves the $1$-clock-synchronization problem. Denote by
$k_{\A'}$ the desired wraparound value of $\A'$, and
let $k_\A = k_{\A'} \cdot \ell$ be the wraparound value for $\A$.

The construction is simple: each time $\A'$ is executed,
it runs $\A$ and returns the clock
value of $\A$ divided by $\ell$ (that is,
$\lfloor{\frac{\F_\A}{\ell}}\rfloor$).
The intuition behind this
construction is straightforward: $\A$ solves the
$\ell$-clock-synchronization problem, thus, the values it returns
are at most $\ell$ apart. Therefore, the values that $\A'$
returns are at most 1 apart from each other.

\subsection{Future Work}
The current paper has a few drawbacks, each of which is
interesting to resolve.

First, is it possible to reduce the
atomicity requirements; that is, can an atomic step be defined as
a single read or a single write (and not as ``read all registers
and write all registers'')?

Second, can the current algorithm be transported into a message
passing model?

Third, can different coin-flipping algorithms
that operate in the asynchronous setting (\ie \cite{167105}) be
used to reduce the exponential convergence time to something more
reasonable? Perhaps even expected constant time?

Fourth, can the
ratio between \byzantine and non-\byzantine nodes be reduced? \Ie can $n>3f$ be achieved?

Fifth, can the problem of asynchronous \byzantine agreement be reduced to the problem of clock synchronization presented in the current work? (This will show that the expected exponential convergence time is as good as is currently known).

Lastly, the building block \enmasse is interesting by itself. It would be interesting
to find a polynomial solution to  \enmasse.

\section{Acknowledgements}
{Michael Ben-Or is the incumbent of the Jean and Helena Alfassa Chair in Computer Science, and he was supported in part by the Israeli Science Foundation (ISF) research grant. Danny Dolev is Incumbent of the Berthold Badler Chair in Computer Science. Danny Dolev was supported in part by the Israeli Science Foundation (ISF) Grant number 0397373. }

\bibliographystyle{plain}
\bibliography{bibliography}
\end{document}